\documentclass[a4paper,onecolumn]{quantumarticle}
\pdfoutput=1

\usepackage[utf8]{inputenc}
\usepackage[T1]{fontenc}
\usepackage[english]{babel}

\usepackage{amsmath}
\usepackage{amssymb}
\usepackage{amsthm}
\usepackage{mathtools}
\usepackage{graphicx}
\usepackage{float}
\usepackage{enumitem}
\usepackage{subcaption}
\usepackage{booktabs}
\usepackage{tcolorbox}
\usepackage{tikz-cd}
\usetikzlibrary{positioning,arrows.meta}
\usepackage[normalem]{ulem}
\usepackage{csquotes}
\usepackage{verbatim}

\usepackage{physics}
\usepackage{hyperref}
\hypersetup{
  colorlinks=true,
  linktoc=page,
  linkcolor=blue,
  citecolor=blue,
  urlcolor=blue,
}

\usepackage[numbers]{natbib}

\theoremstyle{plain}
\newtheorem{theorem}{Theorem}[section]
\newtheorem{lemma}[theorem]{Lemma}
\newtheorem*{lemma*}{Lemma}
\newtheorem{proposition}[theorem]{Proposition}

\theoremstyle{definition}
\newtheorem{definition}[theorem]{Definition}
\newtheorem{example}[theorem]{Example}
\theoremstyle{remark}
\newtheorem{remark}[theorem]{Remark}

\definecolor{darkgreen}{rgb}{0,0.5,0}

\begin{document}

\title{Reduced Quantum-Reference-Frame Channels for Open Quantum Systems}
\author{Paolo Luppi}
\email{paolo.luppi@unimi.it}
\affiliation{Dipartimento di Fisica ``Aldo Pontremoli'', Universit\`a degli Studi di Milano, via~Celoria~16, I-20133 Milan, Italy}
\affiliation{INFN, Sezione di Milano, Via Celoria 16, I-20133 Milan, Italy}
\author{Viktoria Kabel}
\email{vkabel@phys.ethz.ch}
    \affiliation{Institute for Theoretical Physics, ETH Z{\"u}rich, 8093 Z{\"u}rich, Switzerland}
\author{Flaminia Giacomini}
\email{flaminia.giacomini@uniroma2.it}
	\affiliation{Institute for Theoretical Physics, ETH Z{\"u}rich, 8093 Z{\"u}rich, Switzerland}
\affiliation{Dipartimento di Fisica, Universit{\`a} di Roma Tor Vergata and Sezione INFN Roma2, Via della Ricerca Scientifica 1, 00133, Roma, Italy}
\author{Andrea Smirne}
\email{andrea.smirne@unimi.it}
\affiliation{Dipartimento di Fisica ``Aldo Pontremoli'', Universit\`a degli Studi di Milano, via~Celoria~16, I-20133 Milan, Italy}
\affiliation{INFN, Sezione di Milano, Via Celoria 16, I-20133 Milan, Italy}

\maketitle
\begin{abstract}
When reference frames are treated quantum mechanically, the subsystem structure
of quantum systems is no longer absolute, but depends on the choice of the quantum
reference frame (QRF). This raises a basic question: which dynamical properties
are preserved across QRFs, and which depend on the physical reference
used to define the system? We study this question in the general setting of
open quantum systems. At the operational level, after a QRF transformation, the
old reference frame and environmental degrees of freedom may be inaccessible
and must therefore be traced out.
This motivates the definition of \emph{reduced quantum-reference-frame
channels}: maps that connect the joint description in one frame to the
accessible subsystem in another. We characterize their symmetry-constrained
structure and
define a regime in which a reduced
entropy--coherence conservation law holds. We also identify when the induced
reduced action on the open system admits a classical interpretation as random
frame misalignment, and when it instead reflects quantum
reduced-frame effects.
We then apply the framework to pure-dephasing dynamics and derive a necessary and sufficient compatibility condition for population
preservation. When the frame symmetry commutes with the open system's free
Hamiltonian, coherences acquire a multiplicative frame factor, so that locally
inferred decoherence rates split into environmental and reference-induced
contributions. Ramsey interferometry gives this split a direct operational meaning.
Finally, a gravity-motivated dephasing model illustrates how degradation of a
phase reference can mimic signatures usually attributed to intrinsic
decoherence mechanisms.
\end{abstract}

\tableofcontents

\section{Introduction}
\label{sec:intro}

The standard treatment of open quantum systems \cite{Breuer2002,Rivas2011,Vacchini2024} is formulated relative to reference structures that define states, observables, control operations, and measurement outcomes. When these  structures are classical or externally specified, their role in defining the open system dynamics is well understood,
for example in the formulation of symmetry-covariant dynamical maps \cite{Holevo1993,Holevo1996,Vacchini2010,Gasbarri2017,Cattaneo2025}. Once reference systems are themselves treated as quantum objects, however, this separation becomes nontrivial. A basic question then arises: when is the qualitative structure of a reduced dynamics an intrinsic property of the underlying interaction, and when does it depend on the physical reference used to define and probe the system?

Quantum reference frames (QRFs) provide a natural setting in which to address this question. There are several different approaches that model explicitly the quantum properties of the reference frame itself \cite{Bartlett2007,Giacomini2019,Vanrietvelde_2018a, delaHamette_2021_perspectiveneutral, Castro_Ruiz_2021, Loveridge_2016,Carette_2023}. In the following, we focus on the \emph{perspectival approach} to QRFs \cite{Giacomini2019, delaHamette_2020}
\footnote{This approach has been shown to be equivalent to the \emph{perspective-neutral approach} in the case of \emph{ideal} QRFs \cite{Vanrietvelde_2018a, delaHamette_2021_perspectiveneutral}. Here, we restrict our attention to ideal frames.}, in which
 one starts from the description of some external systems relative to one physical system -- the initial reference frame -- and changes into the description relative to another physical system -- the final reference frame -- via a unitary transformation. The reference frame does not describe itself in its own perspective, meaning that it is not assigned a Hilbert space. Hence, the Hilbert spaces of the systems of interest and reference frames correspond to the relational degrees of freedom between the external systems and the reference frame. Under such a change of QRF, properties like superposition, entanglement, and coherence become frame-dependent. In particular, it has been established that the subsystem structure itself depends on the choice of QRF \cite{Giacomini2019, Vanrietvelde_2018a, Castro_Ruiz_2021,hoehn2021quantum,delaHamette_2021_perspectiveneutral,Hoehn2023}. 

While much of the existing literature studies QRF transformations and their consequences at the level of the full relational state, in many realistic settings an observer has operational access only to the system of interest, and not to the other reference frames or to the surrounding environment. After a change of frame, the original reference becomes part of the external degrees of freedom in the new description and, if it is not accessible, it must be discarded together with the environment, leaving an effective reduced description of the system. Along these lines, on the quantum-information side, \cite{Bartlett2006} considers the degradation of a QRF under repeated measurements, which could be understood as an idealized and discrete version of the interaction with an environment. Operational changes between finite-resolution quantum reference frames were shown in ~\cite{ChangingQRF} to induce an effective random-unitary decoherence map in a static, single-frame-change scenario. Within the perspectival approach, \cite{Tuziemski2020} discusses the frame-dependence of decoherence  in the context of a QRF for the Galilei group, while \cite{Le2020} considers its implications for quantum Darwinism. In the context of the perspective-neutral approach, a more general analysis of reduced subsystem dynamics was undertaken in \cite{Hoehn2023}, with a focus on the thermodynamic consequences and, in particular, the frame-dependence of entropy. The latter has received much further attention in the context of conservation laws for QRFs \cite{Cepollaro, delaHamette_entropy_2026} as well as in the context of the observer-dependence of gravitational entropy \cite{Witten_2021,Chandrasekaran_2022,DeVuyst_2024,DeVuyst_2025}. 

Taken together, these works show how different quantities and phenomena, such as reduced states, entropies, thermodynamic quantities, and decoherence, can be frame-dependent. What remains less understood is the structure of the reduced transformation responsible for these effects. In particular, once the original reference frame and environmental degrees of freedom are inaccessible, the relevant object is not the global QRF unitary itself, but the effective channel induced on the accessible subsystem. Since this channel is obtained by a unitary QRF transformation followed by a partial trace, it is not an arbitrary reduced channel: its structure is constrained by the symmetry representation defining the frame change, by the state of the reference, and by the way system, reference, and environment are repartitioned across perspectives. We explicitly define the channel relating the joint description in one reference frame to the reduced open system description in the other, and name it  \emph{reduced QRF channel}. Characterizing its structure allows us to identify which qualitative features of an open-system dynamics are preserved, modified, or induced by a change of QRF.

We first identify a
hierarchy of conditions of the reduced QRF channels that preserve the block structure, associated for example with the energy eigenspaces of the system. Such conditions determine when populations and
coherences remain separated across frames; in particular, a strong form of block-preservation
yields a reduced-level conservation law that provides a channel-theoretic
counterpart of recent coherence--entanglement entropy trade-offs
\cite{Cepollaro,delaHamette_entropy_2026}. Beyond this kinematic level, we
characterize when the induced reduced transformation admits a classical interpretation
as random frame misalignment, and when it displays non-classical
signatures. Since classical misalignment maps are unital, non-unitality of a
transformed reduced dynamics provides an operational witness for non-classical
reduced QRF effects, whenever the original dynamics is known to be unital.

We then turn to pure dephasing, the paradigmatic decoherence model in which
populations are preserved while coherences decay. Pure dephasing is
experimentally relevant in many platforms
\cite{Yoshihara2006,Kwiatkowski2018,Mercurio2023,Yuan2022} and arises naturally
in gravity-motivated decoherence models
\cite{Blencowe2013,Pikovski2015,Bassi2017GravitationalDecoherence,Fahn2026Generalising}. We show that kinematic block-preservation
alone does not suffice to preserve the population structure of pure dephasing
under a frame change. The missing ingredient is a dynamical compatibility
condition between the reduced QRF channel and the controlled evolution
generating dephasing in the original frame, and we prove that this condition is
necessary and sufficient for the preservation of a fixed population
decomposition.

Finally, we show that when the frame-change
symmetry preserves the system's energy structure, transformed coherences
acquire a multiplicative frame factor, and the locally inferred
decoherence rate splits additively into an environmental contribution and a
reference-induced one. Ramsey interferometry gives this decomposition a direct
operational meaning: the same physical interferometer yields frame-invariant
statistics when preparation and measurement are transformed consistently,
whereas local Ramsey experiments tied to different phase standards can infer
different decoherence rates. We illustrate this mechanism in a gravity-motivated
pure-dephasing model with a noisy phase reference, where environment-induced
decoherence can be separated from decreases of visibility due to the degradation of the phase reference.

The remainder of the paper is organized as follows.
Section~\ref{sec:Background} introduces the open-system setting and the QRF framework.
Section~\ref{sec:qrf_preserving} defines and characterizes the structure of reduced QRF channels, including block-preservation conditions and the classical-vs-quantum aspects.
Section~\ref{sec:gpreserving_pd} establishes the dynamical compatibility criterion for the preservation of the population structure of pure-dephasing dynamics.
Section~\ref{sec:Ramsey} gives the operational interpretation in terms of Ramsey interferometry, and Section~\ref{sec:Blencowe} applies the framework to a gravity-motivated pure-dephasing model with a noisy phase reference.
Finally, we conclude in Section~\ref{sec:conclusions} with a discussion of our results and open directions.

\section{Background and Framework}
\label{sec:Background}
\subsection{Open quantum systems}
\label{sec:oqs}

The physical systems we aim to describe are never perfectly isolated in practice \cite{Breuer2002}: they interact with external degrees of freedom that may be inaccessible or only partially controlled. At the same time, QRFs are themselves physical systems, and can therefore interact with the systems under consideration, become correlated with their environments, or possess inaccessible degrees of freedom of their own. To analyze how QRF transformations affect the open-system description, and in particular decoherence processes, we first fix notation and recall the standard open-system framework, with emphasis on pure dephasing.

Consider a quantum system $S$ with Hilbert space $\mathcal H_S$ and free Hamiltonian $H_S$,
coupled to an environment $E$ with Hilbert space $\mathcal H_E$ and Hamiltonian $H_E$
through an interaction Hamiltonian $H_I$ acting on $\mathcal H_S\otimes\mathcal H_E$.
The total Hamiltonian is
\begin{equation}
  H = H_S\otimes\mathbb{I}_E + \mathbb{I}_S \otimes H_E + H_I .
  \label{eq:total_hamiltonian}
\end{equation}
Throughout, we assume $H$ to be time independent and we will often omit the identity
operators $\mathbb{I}_S$ and $\mathbb{I}_E$ when no confusion can arise. In addition, we assume that the system Hamiltonian $H_S$ has a pure-point spectrum,
with a (finite or countably infinite) orthonormal eigenbasis $\{\lvert n\rangle\}$.
Quantum states are represented by statistical operators, i.e. positive, self-adjoint, trace-class operators of unit trace on the corresponding Hilbert space.
\\The joint $S$--$E$ state evolves unitarily as $\rho_{SE}(t)=U(t)\rho_{SE}(0)U^\dagger(t)$
with $U(t)=e^{-iHt/\hbar}$, while the open-system state, also called reduced state, is obtained via the partial trace
over the environmental degrees of freedom, according to
\begin{equation}
  \rho_S(t) = \Tr_E\!\big[\rho_{SE}(t)\big].
\end{equation}
For factorized initial states $\rho_{SE}(0)=\rho_S(0)\otimes\rho_E(0)$, and a fixed initial environmental state $\rho_E(0)$, the
evolution of $S$ is described by a quantum dynamical map $\Phi_t$ that assigns to any initial system state $\rho_S(0)$ the corresponding state at time $t$:
\begin{equation}
  \rho_S(t) = \Phi_t[\rho_S(0)] :=
  \Tr_E\!\Big[U(t)\big(\rho_S(0)\otimes\rho_E(0)\big)U^\dagger(t)\Big].
  \label{eq:dynamical_map}
\end{equation}
Indeed, for given $U(t)$ and $\rho_E(0)$, the same map $\Phi_t$ applies to all initial open-system states, and, by linearity, 
to any trace-class operator on $\mathcal{H}_S$, $\rho_S \in \mathcal T(\mathcal H_S)$.
The map $\Phi_t$ is completely positive and trace-preserving (CPTP)\footnote{Complete positivity ensures positivity under extension by an arbitrary ancillary system; see, e.g., Refs.~\cite{Breuer2002,Benatti2005}.}, i.e., it is a quantum channel.
For a system of finite dimension $d$, we will say that the CPTP map $\Phi_t$ is unital if it preserves the identity operator,
\begin{equation}
\Phi_t[\mathbb I]=\mathbb I;
\end{equation}
equivalently, $\Phi_t$ leaves the maximally mixed state invariant, $\Phi_t(\mathbb I/d)=\mathbb I/d$.

In the following, given a map $\Omega:$ $\mathcal{T}(\mathcal H) \rightarrow \mathcal{T}(\mathcal H')$, with $\mathcal{H}$ 
and $\mathcal{H}'$ two possibly different Hilbert spaces, we will often refer to its dual map 
$\Omega^\dagger:$ $\mathcal{B}(\mathcal{H}') \rightarrow\mathcal{B}(\mathcal{H})$ -- with $\mathcal{B}(\mathcal{H})$ 
the set of bounded linear operators on $\mathcal{H}$ -- via the relation
\begin{equation}
\label{eq:dualmap}
    \Tr[X\,\Omega[\rho]]=\Tr[\Omega^\dagger[X]\,\rho] \quad 
    \forall \rho \in \mathcal{T}(\mathcal{H})\,\,\forall X \in \mathcal{B}(\mathcal{H}').
\end{equation}

\subsubsection*{Pure dephasing}
\label{subsec:pure_dephasing}
Pure dephasing is a paradigmatic model of decoherence, since it corresponds to coherence loss without population transfer \cite{Breuer2002,Roszak2019EntanglementObjectivity}.
Given the total Hamiltonian in Eq.(\ref{eq:total_hamiltonian}),
let $H_S=\sum_n \lambda_n P_n$ be the spectral decomposition of the free system Hamiltonian into orthogonal
projectors $P_n$ (of rank $d_n$; in the non-degenerate case $P_n=\ket n\!\bra n$).
We say that the reduced dynamics is of \emph{pure-dephasing type} in the energy
basis if $H_S$ is invariant under the global Hamiltonian evolution, that is
\begin{equation}
  [H_S,H]=[H_S,H_I]=0.
  \label{eq:microscopic_condition}
\end{equation}
In the non-degenerate case, this means that the global unitary admits a controlled form
\begin{equation}
  U(t)=\sum_n P_n\otimes U_n(t),
  \label{eq:controlled_unitary}
\end{equation}
where each $U_n(t)$ acts on $E$ conditioned on $S$ being in the $n$-th energy
block.\footnote{For degenerate levels, pure dephasing preserves block populations but may allow nontrivial intra-block dynamics.}
Since the global evolution preserves each eigenspace of $H_S$, it immediately
follows that the energy-block populations are time-independent for every initial
state,
\begin{equation}
  p_n(t):=\Tr[P_n\,\rho_S(t)]=\Tr[P_n\,\rho_S(0)]= p_n(0)
  \quad \forall n,\ t\ge 0,\ \forall \rho_S(0),
  \label{eq:populations_constant}
\end{equation}  
or, equivalently, that the energy projectors are fixed points of the dual reduced dynamical map,  
\begin{equation}
  \Phi_t^\dagger[P_n]=P_n \quad \forall n,\ \forall t\ge 0;
  \label{eq:pure_dephasing_heisenberg}
\end{equation}
moreover, off-diagonal terms (coherences) decay multiplicatively, which in the non-degenerate case simply reads
\begin{equation}
  \rho_{nm}(t)=\bra n\rho_S(t)\ket m = f_{nm}(t)\,\rho_{nm}(0)\quad (n\neq m),
  \label{eq:coherence_decay}
\end{equation}
with the coherence factors 
\begin{equation}
    f_{nm}(t)= \Tr_E[U_n(t)\rho_E(0)U_m^{\dagger}(t)]
\end{equation}
that satisfy $f_{nm}(0)=1$, $f_{mn}(t)=f_{nm}^*(t)$ and
$|f_{nm}(t)|\le 1$. Indeed, the decay of the coherences can be characterized by the (generally time-dependent) decoherence rates
\begin{equation}
  \gamma_{nm}(t):=-\frac{d}{dt}\ln|f_{nm}(t)|
  \ =\ -\frac{d}{dt}\ln|\rho_{nm}(t)|,
  \quad n\neq m.
  \label{eq:decoherence_rate}
\end{equation}

\subsection{Quantum reference frames}
\label{subsec:qrf_setup}

The main idea behind QRFs is to treat the reference systems as quantum systems. Reference frames play an important role whenever the theory under consideration has a symmetry. One of the simplest examples is translation symmetry: in a translationally invariant theory, only relative positions are physically meaningful. Choosing a reference frame then amounts to identifying the origin of the coordinate system with a physical system, so that the coordinates of all other systems are defined relative to it. When studying QRFs, one treats this reference system quantum mechanically and, in particular, takes into account the possibility of different reference frames to be in superposition relative to one another or entangled with the other systems.

This is relevant in quantum information and quantum foundations \cite{Aharonov1967,AharonovKaufherr1984,Bartlett2006,Bartlett2007,Angelo2011,Angelo2012}, where reference systems such as clocks, phase standards, or orientation frames may need to be treated as physical quantum systems rather than as ideal external backgrounds. It becomes especially natural in relational or background-free settings, and in particular in quantum gravity, where one cannot in general rely on a fixed external spacetime structure and physical quantities must be defined relative to other systems (e.g.~\cite{Rovelli1991, Tambornino2011}).

Several approaches to QRFs have been developed in the literature. The modern approaches include in particular the perspectival \cite{Giacomini2019, delaHamette_2020}, perspective-neutral \cite{Vanrietvelde_2018a, delaHamette_2021_perspectiveneutral}, extra-particle \cite{Castro_Ruiz_2021,Garmier_2025}, and operational approach \cite{Loveridge_2016,Carette_2023}. In this work we adopt the perspectival approach of Refs.~\cite{Giacomini2019,delaHamette_2020}, which is well suited to the reduced-channel viewpoint developed here.

Let us now briefly review the formalism of QRFs in the perspectival approach. To simplify the construction, we focus on the case of three quantum systems: $A$ and $B$, which will take the role of the reference frames, and an additional system $S$. Moreover, we assume that the symmetry is given by a locally compact group $G$.\footnote{Local compactness ensures the existence
of a left Haar measure $dg$ on $G$, which is needed to define the integral form
of the QRF transformation \cite{delaHamette_2020}.} In the perspectival approach, one is only ever concerned with the perspectives relative to a given reference frame. We thus start with the description relative to system $A$. This description contains only the systems external to $A$: the quantum state of $S$ and $B$ relative to $A$ is given by a state
\begin{equation}
\ket{\psi}^{(A)}_{SB} \in \mathcal H^{(A)}_{SB}
=\mathcal H_S^{(A)}\otimes\mathcal H_B^{(A)},
\end{equation}
where the superscript $(A)$ labels the perspective. The group $G$ acts on the Hilbert space $\mathcal H^{(A)}_{SB}$ via the unitary representation $V_S(g)\otimes V_B(g)$. We are assuming \emph{ideal QRFs}, which have Hilbert space $\mathcal H_B^{(A)} = \mathcal H_A^{(B)} = L^2(G)$ and carry the left regular representation of $G$. This ensures that the basis states of the QRF $\{\ket{g}\}_{g\in G}$ are orthogonal, $\braket{h}{g}=\delta(h^{-1}g)$, and can thus be distinguished perfectly.

The perspectives relative to two different QRFs are related by a unitary transformation. In order to change into the description relative to $B$, we apply the \emph{QRF transformation} \cite{delaHamette_2020}
\begin{equation}
    \hat{S}_{A \to B}= \int_G dg\, \ket{g^{-1}}_A\!\bra{g}_B\otimes V^{\dagger}_{S}(g),\label{eq:QRF_transformation}
\end{equation}
where $dg$ is the Haar measure of the group $G$. This defines a unitary map between \emph{different} perspectival Hilbert spaces,
\begin{align}
    \hat{\mathcal S}_{A\to B}: \mathcal H^{(A)}_{SB} \longrightarrow \mathcal H^{(B)}_{SA} = \mathcal H_S^{(B)}\otimes\mathcal H_A^{(B)}.
\end{align}

The state of $A$ and $S$ relative to $B$ is thus given by $\ket{\psi}^{(B)}_{SA}=\hat{S}_{A\to B}\ket{\psi}^{(A)}_{SB}$. The unitarity of the transformation ensures that observable quantities, such as probabilities and expectation values, are conserved. However, the transformation $\hat{S}_{A\to B}$ acts as a change of basis in the total Hilbert space, which can lead to a different partitioning into subsystems. Formally this means that, while the Hilbert spaces $\mathcal H^{(A)}_{SB}$ and $\mathcal H^{(B)}_{SA}$ are isomorphic, they are, in general, equipped with different tensor-product structures \cite{hoehn2021quantum,Hoehn2023}. This \emph{subsystem relativity} under changes of QRF has important implications; most notably, it leads to the \emph{frame-dependence of entanglement}.

\begin{example}[$\mathbb{Z}_2$-Symmetry]
As a simple example, consider a system of two qubits $S$ and $B$ relative to $A$ and the discrete symmetry group $G=\mathbb{Z}_2$ \cite{delaHamette_2020, Cepollaro}. In this case, the QRF transformation between $A$ and $B$ is given by
\begin{equation}
\hat{S}_{A\to B}=\ket0_A\!\bra0_B\otimes\mathbb I_S+\ket1_A\!\bra1_B\otimes\sigma_x,
\label{eq:controlled_sigma_x}
\end{equation}
which is the discrete version of Eq.~\eqref{eq:QRF_transformation} with $V_S(e)=\mathbb I$ and $V_S(x)=\sigma_x$. Consider as a concrete example the following state of $S$ and $B$ relative to $A$:
\begin{equation}
    \ket{\psi}^{(A)}_{SB} = \frac{1}{\sqrt{2}}\ket{0}_S\left(\ket{0}_B+ \ket{1}_B\right).
\end{equation}
Applying the QRF transformation $\hat{S}_{A\to B}$, we obtain the state
\begin{equation}
    \ket{\psi}^{(B)}_{SA} = \frac{1}{\sqrt{2}}\left(\ket{00}_{SA}+ \ket{11}_{SA}\right)
\end{equation}
of $S$ and $A$ relative to $B$. The QRF transformation re-orients the spin of $S$, depending on the orientation of the new reference frame $B$. Note that while the state relative to $A$ is a product state, the state relative to $B$ is entangled.
\end{example}

More generally, subsystem relativity implies that properties often regarded as intrinsic to a subsystem --- including coherence, entanglement entropy, decoherence behavior, thermodynamic quantities such as heat and work, and even the distinction between open and closed subsystems --- may become frame dependent under a QRF transformation \cite{Cepollaro,DeVuyst_2024,Tuziemski2020,Hoehn2023}. This naturally raises the question of how the qualitative structure of reduced open-system dynamics transforms under a change of quantum reference frame.

\section{Reduced QRF channels: definition, structure, and properties}
\label{sec:qrf_preserving}

We now introduce the central object of this work: the \emph{reduced QRF channel}. Once inaccessible degrees of freedom, associated to reference systems or the environment, are discarded, a change of QRF induces a reduced CPTP map to the accessible system. 
The resulting structure determines which qualitative features of open-system dynamics are preserved under a change of frame. In this section we identify those constraints, distinguish classically interpretable reduced transformations from quantum ones, and characterize the population--coherence structures that reduced QRF channels can preserve.

To capture a typical open-quantum-systems setting, we include an environment $E$ in the perspective of $A$, in addition to the system $S$ and the alternative reference frame $B$. The total Hilbert space is therefore
\begin{equation}
\mathcal H^{(A)}_{SBE}
=
\mathcal H_S^{(A)}\otimes \mathcal H_B^{(A)}\otimes \mathcal H_E^{(A)},
\end{equation}
and the joint state evolves unitarily according to
\begin{equation}
\rho^{(A)}_{SBE}(t)
=
U^{(A)}_{SBE}(t)\,\rho^{(A)}_{SBE}(0)\,U^{(A)\dagger}_{SBE}(t).
\label{eq:global_evolution_A}
\end{equation}
In the presence of the environment, we use the natural extension of the QRF transformation to the joint degrees of freedom $S$ and $E$, so that
\begin{equation}\label{eq:aux1}
\hat S_{A\to B}
=
\int_G dg\, \ket{g^{-1}}_A\!\bra{g}_B \otimes V^{\dagger}_{SE}(g),
\end{equation}
where $V_{SE}(g)$ is a unitary representation on
$\mathcal H_S^{(A)}\otimes\mathcal H_E^{(A)}$. In the following, we take the QRF transformation \(\hat S_{A\to B}\) to be time independent. Note, however, that the state $\rho_B(t)$ of the reference frame may still evolve in time and thus lead to a non-trivial time dependence on the reduced level.\footnote{\label{fn:time-dependent-QRF-trafo}An example for an explicitly time dependent QRF transformation is the transformation for the extended Galilei group \cite{Giacomini2019}. More generally, a time-dependent
transformation \(\hat S_{A\to B}(t)\) would define a family of reduced channels
\(\mathcal Q_t[\rho]=
\Tr_{AE}[\hat S_{A\to B}(t)\rho\hat S_{A\to B}^\dagger(t)]\). Such a case is briefly discussed below.}
Applying this transformation yields the state relative to $B$,
\begin{equation}
\rho^{(B)}_{SAE}(t)
=
\hat S_{A\to B}\,\rho^{(A)}_{SBE}(t)\,\hat S_{A\to B}^\dagger .
\label{eq:state_transformation}
\end{equation}
If, in the reference frame of $B$, only observables on $S$ are accessible, the operationally relevant state is obtained by tracing out the inaccessible degrees of freedom $(A,E)$. This motivates the following definition.

\begin{definition}[Reduced QRF channel]
\label{def:reduced_qrf_channel}
Let $\hat S_{A\to B}:\mathcal H^{(A)}_{SBE}\to \mathcal H^{(B)}_{SAE}$ be the unitary QRF transformation from the perspective of $A$ to that of $B$ defined in Eq.(\ref{eq:aux1}). The \emph{reduced QRF channel} is the map
\begin{equation}
\mathcal Q:\mathcal T(\mathcal H^{(A)}_{SBE})\to \mathcal T(\mathcal H^{(B)}_S),
\qquad
\mathcal Q[\rho]
:=
\Tr_{AE}\!\big[\hat S_{A\to B}\,\rho\,\hat S_{A\to B}^\dagger\big].
\label{eq:reduced_quantumchannel_rewrite}
\end{equation}
\end{definition}
Since $\mathcal Q$ is the composition of a unitary conjugation and a partial trace, it is CPTP. Crucially, however, $\mathcal Q$ is \emph{not} by itself a dynamical map on $S$ alone: it acts on states on the joint degrees of freedom $SBE$ in the $A$-description and returns the reduced system state in the $B$-description. In particular,
\begin{equation}
\rho_S^{(B)}(t)
=
\mathcal Q\!\big[\rho_{SBE}^{(A)}(t)\big]
=
\Tr_{AE}\!\big[\hat S_{A\to B}\,\rho_{SBE}^{(A)}(t)\,\hat S_{A\to B}^\dagger\big].
\label{eq:system_state_frame_B_rewrite}
\end{equation}
A system map on $S$ alone arises only after fixing a factorized initial
preparation in frame $A$,
$\rho_{SBE}^{(A)}(0)=\rho_S^{(A)}(0)\otimes \rho_{BE}^{(A)}(0)$.
In that case, combining $\mathcal Q$ with the joint unitary dynamics induces a
CPTP map
$\Lambda_t:\mathcal T(\mathcal H_S^{(A)})\to
\mathcal T(\mathcal H_S^{(B)})$
such that
$\rho_S^{(B)}(t)=\Lambda_t(\rho_S^{(A)}(0))$.
The relation between $\mathcal Q$ and $\Lambda_t$ is summarized in the following diagram
\begin{equation}
\centering
\begin{tikzcd}[column sep=huge, row sep=huge]
\rho_{SBE}^{(A)}(0)
  \arrow[r, "{\mathcal{U}^{(A)}_{SBE}(t)}"]
  \arrow[d, "{\mathrm{Tr}_{BE}}"'] &
\rho_{SBE}^{(A)}(t)
  \arrow[d, "{\mathcal{Q}}"] \\
\rho_{S}^{(A)}(0)
  \arrow[r,dashed, "{\Lambda_t}"'] &
\rho_{S}^{(B)}(t)
\end{tikzcd}
\label{fig:qrf_commutative}
\end{equation}
where the bottom arrow is defined after fixing the factorized preparation
$\rho_{SBE}^{(A)}(0)=
\rho_S^{(A)}(0)\otimes\rho_{BE}^{(A)}(0)$,
with fixed \(\rho_{BE}^{(A)}(0)\).
Thus, \(\mathcal Q\) is the preparation-independent channel associated with the
reduced frame change, whereas \(\Lambda_t\) is a preparation-dependent induced
open-system map. Since \(\Lambda_t\) is not defined canonically without a choice
of initial assignment, we formulate the structural constraints at the level of
\(\mathcal Q\) and use \(\Lambda_t\) only after such an assignment has been fixed.

\begin{example}
\label{ex:reducedChannel}
    Let $B$ be a qubit reference frame with pointer basis $\{\ket{0}_B,\ket{1}_B\}$, and consider the discrete $\mathbb{Z}_2$ QRF transformation
\begin{equation}
\hat S_{A\to B}
=
\ket{0}_A\!\bra{0}_B \otimes \mathbb{I}_S \otimes \mathbb{I}_E
+
\ket{1}_A\!\bra{1}_B \otimes \sigma_x^{(S)} \otimes V^{\dagger}_E,
\label{eq:z2_qrf_transformation_block_example}
\end{equation}
where $V_E$ is an arbitrary unitary acting on $E$ such that
$V_E^2=\mathbb I_E$. This is the $\mathbb{Z}_2$
instance of a factorized representation on $SE$, with
\begin{equation}
V^{\dagger}_{SE}(0)=\mathbb{I}_S\otimes\mathbb{I}_E,
\qquad
V^{\dagger}_{SE}(1)=\sigma_x^{(S)}\otimes V^{\dagger}_E.
\end{equation}
Assume that, relative to $A$, the state of $S$, $B$, and $E$ evolves as a product state: $\rho_{SBE}^{(A)}(t) = \rho_{SE}^{(A)}(t)\otimes \rho_B^{(A)}(t)$, i.e., there is no interaction between $B$ and $SE$. In this case, the reduced state relative to $B$, obtained by applying the reduced QRF channel $\mathcal{Q}$, is
\begin{align}
    \rho_S^{(B)}(t) = \mathcal{Q}[\rho_{SBE}^{(A)}(t)] &= \underset{p_0(t)}{\underbrace{\expval{\rho_B^{(A)}(t)}{0}}} \mathrm{Tr}_E[\rho_{SE}^{(A)}(t)]+\underset{p_1(t)}{\underbrace{\expval{\rho_B^{(A)}(t)}{1}}}\mathrm{Tr}_E[(\sigma_x\otimes V^{\dagger}_E)\rho_{SE}^{(A)}(t)(\sigma_x\otimes V_E)]\nonumber\\
    &= p_0(t) \rho_S^{(A)}(t) + p_1(t)\sigma_x\rho_S^{(A)}(t)\sigma_x.
\end{align}
Note that this is an example of a QRF transformation, which can be understood in terms of a classical mixture of reference frame orientations at the reduced level. We discuss the general conditions for this behavior in Sec.~\ref{subsec:Quantumvsclassical} below.
\end{example}

\subsection{Preserving the population--coherence structure}
\label{sec:block-preserving}
We now investigate how block structures associated, for example, with energy eigenspaces, pointer
sectors, or symmetry sectors are affected by a reduced QRF transformation, characterizing in particular the conditions under which the map $\mathcal{Q}$ preserves a given block structure of the system $S$.

We fix a projective decomposition (PVM) on the system Hilbert space in frame~$A$,
\begin{equation}\label{eq:PVM-A}
  \{P_m^{(A)}\}_m \subset \mathcal B(\mathcal H_S^{(A)}),
  \qquad
(P_m^{(A)})^{\dagger} = P_m^{(A)}
  \qquad
  P_m^{(A)} P_{m'}^{(A)}=\delta_{mm'}P_m^{(A)},
  \qquad
  \sum_m P_m^{(A)}=\mathbb I.
\end{equation}
Analogously, we fix a PVM
$\{\widetilde P_n^{(B)}\}_n$ on $\mathcal H_S^{(B)}$. The preservation
condition is formulated directly at the level of the reduced QRF channel.

\begin{definition}[Block-preserving reduced QRF channel]
\label{def:block_preserving}
Fix a PVM $\{P_m^{(A)}\}_m$ on $\mathcal H_S^{(A)}$ and a PVM
$\{\widetilde P_n^{(B)}\}_n$ on $\mathcal H_S^{(B)}$. Let $\mathcal Q$ be a
reduced QRF channel and let
$\mathcal Q^\dagger$ denote its dual map. We say that $\mathcal Q$ is
\emph{block-preserving} with respect to these two PVMs if, for every $n$,
\begin{equation}
\mathcal Q^\dagger\!\big[\widetilde P_n^{(B)}\big]
= \sum_m P_m^{(A)} \otimes \Pi_{nm}^{(A)},
\label{eq:block_preserving_def_only}
\end{equation}
where the operators
$\{\Pi_{nm}^{(A)}\}_{n,m}\subset \mathcal B(\mathcal H_{BE}^{(A)})$
satisfy
\begin{equation}
\Pi_{nm}^{(A)}\ge 0 \quad \forall n,m,
\qquad
\sum_n \Pi_{nm}^{(A)}=\mathbb I_{BE}^{(A)} \quad \forall m .
\label{eq:block_preserving_povm_only}
\end{equation}
Equivalently, for each input block $m$, the family
$\{\Pi_{nm}^{(A)}\}_n$ forms a POVM on the inaccessible degrees of freedom
$BE$.
\end{definition}
At the formal level, the two PVMs entering
Definition~\ref{def:block_preserving} are arbitrary: block-preservation is a
relation between a reduced QRF channel and two chosen block decompositions, one
in each frame. In the physical scenarios of interest, however, these decompositions are
typically not chosen independently. Rather, they are selected by the same
physical or operational criterion in the two descriptions, such as an energy
decomposition, a pointer observable, or a symmetry-sector decomposition. For
example, if a meaningful local energy observable is available in both frames,
one may take $\{P_m^{(A)}\}_m$ and
$\{\widetilde P_n^{(B)}\}_n$ to be the corresponding spectral projectors.
More generally, because a QRF transformation can reshuffle the tensor-product
structure, the transformed Hamiltonian need not determine a canonical local
Hamiltonian acting on $S$ in the new frame. The PVM
$\{\widetilde P_n^{(B)}\}_n$ should therefore be understood in general as a
physically or operationally specified block decomposition in the
$B$-description.

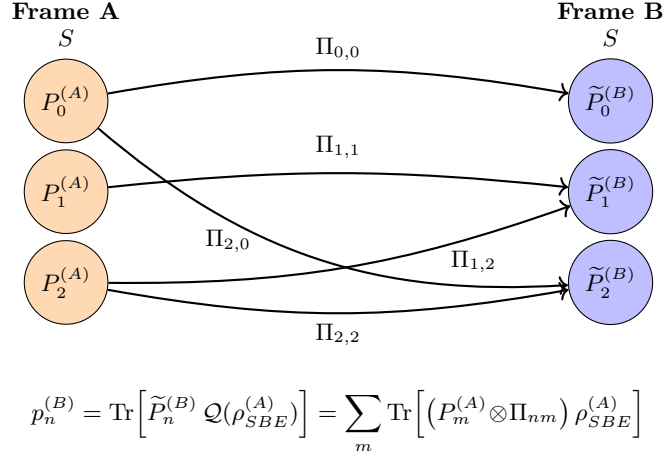
\begin{figure}[H]
\centering
\begin{tikzpicture}[scale=1.2, every node/.style={font=\small}]

\node at (-3,3) {\textbf{Frame A}};
\node at (-3,2.7) {$S$};

\foreach \m/\y in {0/2, 1/1, 2/0} {
  \node[draw,circle,fill=orange!30] (A\m) at (-3,\y) {$P^{(A)}_{\m}$};
}

\node at (3,3) {\textbf{Frame B}};
\node at (3,2.7) {$S$};

\foreach \n/\y in {0/2, 1/1, 2/0} {
  \node[draw,circle,fill=blue!25] (B\n) at (3,\y) {$\widetilde P^{(B)}_{\n}$};}

\draw[->,thick] (A0) to[bend left=10]  node[pos=0.5,above]      {$\Pi_{0,0}$} (B0);
\draw[->,thick] (A0) to[bend right=22] node[pos=0.3,below]      {$\Pi_{2,0}$} (B2);

\draw[->,thick] (A1) to[bend left=6]
      node[pos=0.5, above=2pt] {$\Pi_{1,1}$} (B1);

\draw[->,thick] (A2) to[bend right=10]
      node[pos=0.8, below=1pt] {$\Pi_{1,2}$} (B1);
\draw[->,thick] (A2) to[bend right=10] node[pos=0.5,below]      {$\Pi_{2,2}$} (B2);

\node[align=center] at (0,-1.55) {
  \(\displaystyle p^{(B)}_{n}
   =\mathrm{Tr}\!\left[\widetilde P^{(B)}_{n}\, \mathcal{Q}(\rho^{(A)}_{SBE})\right]
   =\sum_{m}\mathrm{Tr}\!\left[\bigl(P^{(A)}_{m}\!\otimes\!\Pi_{nm}\bigr)\,\rho^{(A)}_{SBE}\right]\)\\  
};
\end{tikzpicture}
\caption{\textit{Block preservation. A block population \(p_n^{(B)}\) in frame \(B\)
may receive contributions from different blocks \(m\) of \(S\) in frame \(A\),
with the dependence on inaccessible degrees of freedom encoded in the effects
\(\Pi_{nm}\). Inter-block coherences of \(S\) in frame \(A\) do not contribute
to these populations. A one-to-one permutation correspondence between input and
output blocks gives the stronger condition of strong block-preservation. Edges
are illustrative.}
}
\label{fig:block_preserving}
\end{figure}

Equation~\eqref{eq:block_preserving_def_only} states that each population projector $\widetilde{P}_n^{(B)}$ in frame $B$ is mapped by $\mathcal Q^\dagger$ to an operator that is
block-diagonal with respect to $\{P_m^{(A)}\}$ while the operators $\Pi_{nm}^{(A)}$ encode a potential dependence on the inaccessible degrees of freedom, $B$ and $E$. As a consequence, block populations in the $B$-description are insensitive to
inter-block coherences of $S$ in the $A$-description.  This suggests a simple operational reading: the block populations of $S$ relative to $B$ are unaffected if one erases all the inter-block coherences of $S$ in the $A$-description, as long as one leaves $B$ and $E$ untouched. 

 We can formalize this in terms of pinching (full-dephasing) maps associated with the relevant projective decompositions
\cite{Hayashi2017}.

\begin{definition}[Pinching (full-dephasing) maps]\label{def:pinching}
Given a PVM $\{\widetilde{P}_n^{(B)}\}_n$ on $\mathcal H_S^{(B)}$ and a PVM
$\{P_m^{(A)}\}_m$ on $\mathcal H_S^{(A)}$, define
\begin{equation}
\widetilde \Delta_{S}^{(B)}:\mathcal T(\mathcal H_S^{(B)})\to\mathcal T(\mathcal H_S^{(B)}),
\qquad
\widetilde\Delta_S^{(B)}(Y):=\sum_n \widetilde P_n^{(B)}\,Y\,\widetilde P_n^{(B)},
\end{equation}
and
\begin{equation}
\Delta_{SBE}^{(A)}:\mathcal T(\mathcal H^{(A)}_{SBE})\to\mathcal T(\mathcal H^{(A)}_{SBE}),
\qquad
\Delta_{SBE}^{(A)}(X):=\sum_m (P^{(A)}_m\!\otimes\!\mathbb I_{BE})\,X\,(P^{(A)}_m\!\otimes\!\mathbb I_{BE}).
\end{equation}
Both maps are completely positive, unital, and idempotent; they cancel coherences between different blocks
defined, respectively, by $\{\widetilde P_n^{(B)}\}$ and by $\{P_m^{(A)}\otimes\mathbb I_{BE}\}$.
\end{definition}

We can now state a necessary and sufficient condition for block preservation in the non-degenerate case, together with a stronger sufficient condition.
\begin{lemma}[Pinching form of block-preservation, rank-$1$ case]\label{lem:pinching}
Assume that both projective decompositions are non-degenerate, i.e.
\(\operatorname{rank} P_m^{(A)}=\operatorname{rank}\widetilde P_n^{(B)}=1\)
for all \(m,n\). Then $\mathcal Q$ is block-preserving  with respect to
\(\{P_m^{(A)}\}_m\) and \(\{\widetilde P_n^{(B)}\}_n\) if and only if 
\begin{equation}
\widetilde\Delta^{(B)}_{S}\circ\mathcal Q
=
\widetilde\Delta^{(B)}_{S}\circ\mathcal Q\circ\Delta_{SBE}^{(A)}.
\label{eq:diagonal_pinching_intertwining}
\end{equation}
Moreover, still assuming rank 1 of the projectors, if
\begin{equation}
\widetilde\Delta^{(B)}_{S}\circ\mathcal Q
=
\mathcal Q\circ\Delta_{SBE}^{(A)}
\label{eq:full_pinching_intertwining}
\end{equation}
then $\mathcal Q$ is block-preserving.
\end{lemma}

The proof of this lemma is provided in Appendix \ref{App:equivalencepinching}.

\vspace{0.3cm}
While this ensures that, in the \(B\)-description, the block populations of \(S\) do not depend on the inter-block coherences of \(S\) relative to \(A\), block-preservation still allows, in general, for a redistribution of populations across blocks. This can be seen on the right-hand side of \eqref{eq:block_preserving_def_only}, which may
involve contributions from several input blocks $m$. We can single out a stronger preservation property that excludes such a redistribution.

\begin{definition}[Strongly block-preserving reduced QRF channel]
\label{def:strong_block_preserving}

Let $\{P_m^{(A)}\}_m$ be a PVM on $\mathcal H_S^{(A)}$, $\{\widetilde P_n^{(B)}\}_n$ a PVM on $\mathcal H_S^{(B)}$, 
$\mathcal Q$ a
reduced QRF channel and
$\mathcal Q^\dagger$ its dual.
We say that
$\mathcal Q$ is \emph{strongly block-preserving} if there exists a permutation
$\pi$ such that
\begin{equation}
\mathcal Q^\dagger\!\big[\widetilde P_n^{(B)}\big]
=
P_{\pi(n)}^{(A)}\otimes \mathbb I_{BE}
\qquad \forall n .
\label{eq:strong_block}
\end{equation}
\end{definition}

Strong block-preservation implies that each output
block is associated with a unique input block,
and that
the corresponding population is preserved, independently of the state of the
inaccessible degrees of freedom. 
Indeed, strong block-preservation implies block-preservation, but not conversely,
as shown by the example below. Thus, strongly block-preserving reduced QRF
channels form a proper subclass of block-preserving ones.
In the factorized case \(V_{SE}(g)=V_S(g)\otimes V_E(g)\), a simple sufficient
condition for strong block-preservation, when the same block decomposition is
used in both frames, is $[P_m^{(A)},V_S(g)]=0
\,\,\,\, \forall \,m,g $.

\begin{example}[$\mathbb{Z}_2$ channel: block-preserving but not strongly block-preserving]
\label{ex:z2_block_preserving_not_strong}
Let $S$ be a qubit with Hamiltonian
\begin{equation}
H_S^{(A)}=\frac{\omega}{2}\sigma_z,
\end{equation}
and energy projectors
$P_0=\ket{0}\!\bra{0}$ and $P_1=\ket{1}\!\bra{1}$.
Let $B$ be a qubit reference frame with basis
$\{\ket{0}_B,\ket{1}_B\}$, and consider the discrete $\mathbb{Z}_2$ QRF
transformation
\begin{equation}
\hat S_{A\to B}
=
\ket{0}_A\!\bra{0}_B \otimes \mathbb{I}_S \otimes \mathbb{I}_E
+
\ket{1}_A\!\bra{1}_B \otimes \sigma_x^{(S)} \otimes V^{\dagger}_E,
\label{eq:z2_qrf_transformation_block_example2}
\end{equation}
from Example \ref{ex:reducedChannel}.
We test block-preservation with respect to the energy decomposition specified
in both frame descriptions. In frame $A$, this decomposition is given by the
spectral projectors $P_n$ of $H_S^{(A)}$. For the present qubit example, we identify the corresponding population projectors in frame $B$ with
the same two-outcome PVM,
\begin{equation}
\widetilde P_n^{(B)}=P^{(A)}_n,
\qquad
n=0,1.
\end{equation}
 We then ask
whether the populations associated with this PVM in frame $B$ are insensitive to
the inter-block energy coherences of $S$ in frame $A$. The dual reduced QRF channel acts as
\begin{equation}
\mathcal Q^\dagger(X)
=
\hat S_{A\to B}^{\dagger}
\bigl( X\otimes\mathbb{I}_A\otimes\mathbb{I}_E \bigr)
\hat S_{A\to B}.
\label{eq:z2_dual_reduced_qrf}
\end{equation}
Since $\sigma_x P_0\sigma_x=P_1$ and $
\sigma_x P_1\sigma_x=P_0$,
we obtain
\begin{align}
\mathcal Q^\dagger(\widetilde P^{(B)}_0)
&=
P^{(A)}_0\otimes\ket{0}\!\bra{0}_B\otimes\mathbb{I}_E
+
P^{(A)}_1\otimes\ket{1}\!\bra{1}_B\otimes\mathbb{I}_E,
\label{eq:z2_dual_P0}\\
\mathcal Q^\dagger(\widetilde P^{(B)}_1)
&=
P^{(A)}_0\otimes\ket{1}\!\bra{1}_B\otimes\mathbb{I}_E
+
P^{(A)}_1\otimes\ket{0}\!\bra{0}_B\otimes\mathbb{I}_E.
\label{eq:z2_dual_P1}
\end{align}
Therefore $\mathcal Q$ is block-preserving, with
\begin{align}
\Pi^{(A)}_{00} &= \ket{0}\!\bra{0}_B\otimes\mathbb{I}_E,
&
\Pi^{(A)}_{01} &= \ket{1}\!\bra{1}_B\otimes\mathbb{I}_E,
\nonumber\\
\Pi^{(A)}_{10} &= \ket{1}\!\bra{1}_B\otimes\mathbb{I}_E,
&
\Pi^{(A)}_{11} &= \ket{0}\!\bra{0}_B\otimes\mathbb{I}_E.
\label{eq:z2_block_preserving_effects}
\end{align}
For each input block $m$, the effects satisfy
\begin{equation}
\Pi_{0m}+\Pi_{1m}
=
\mathbb{I}_{BE},
\qquad
m=0,1,
\end{equation}
and hence define a POVM on $BE$.

The channel is not strongly block-preserving. Indeed, strong block-preservation
would require a permutation $\pi$ such that
\begin{equation}
\mathcal Q^\dagger(\widetilde P^{(B)}_n)
=
P^{(A)}_{\pi(n)}\otimes\mathbb{I}_{BE}
\qquad
\forall n.
\end{equation}
However, both $\mathcal Q^\dagger(\widetilde P^{(B)}_0)$ and $\mathcal Q^\dagger(\widetilde P^{(B)}_1)$ contain
contributions from the two distinct input blocks, correlated with different
effects on $B$. Consequently, neither operator can be written as a single system
projector tensored with $\mathbb{I}_{BE}$. This example therefore shows that
block-preservation does not imply strong block-preservation.
\end{example}

\subsection{A conservation law under strong block-preservation}
\label{subsec:strong_block_conservation}

Recent works \cite{Cepollaro,delaHamette_entropy_2026} have established conservation laws trading subsystem coherence against entanglement under QRF transformations. Here we show that an analogous invariant emerges already at the reduced-system level, as a consequence of the reduced QRF channel under the assumption of strong
block-preservation.

Let us assume that a system $S$, a reference frame $B$, and an environment $E$ are in an arbitrary state $\rho_{SBE}^{(A)}$ relative to $A$, with $\rho_S^{(A)}=\Tr_{BE}[\rho_{SBE}^{(A)}]$ the reduced state of the system in this frame. Using the von Neumann entropy $\mathcal S[\rho]:=-\Tr[\rho\log\rho]$, we define
the local entropy of $S$  in the $A$-description as
\begin{equation}
  \mathcal E_S^{(A)} := \mathcal S[\rho_S^{(A)}].
\end{equation}
Moreover, using the pinching map $\Delta_S^{(A)}$ associated
with the spectral PVM $\{P_n^{(A)}\}_n$ of $H_S^{(A)}$ (Def.~\ref{def:pinching}),
we define the dephased state
\begin{equation}
  \rho_{S,d}^{(A)} := \Delta_S^{(A)}[\rho_S^{(A)}]
  = \sum_n P_n^{(A)}\,\rho_S^{(A)}\,P_n^{(A)},
\end{equation}
and the relative entropy of coherence of $S$ in that energy basis as \cite{Baumgratz2014}
\begin{equation}
  \mathcal C_e^{(A)} := \mathcal S[\rho_{S,d}^{(A)}]-\mathcal S[\rho_S^{(A)}],
\end{equation}
which quantifies the amount of coherence of $\rho_S^{(A)}$
with respect to $\{P_n^{(A)}\}_n$.

Similarly, in the $B$-description we have $\rho_S^{(B)} =\mathcal Q[\rho_{SBE}^{(A)}]$,
where $\mathcal Q$ is the reduced QRF channel from $A$ to $B$, and we define
\begin{equation}
  \mathcal E_S^{(B)} := \mathcal S[\rho_S^{(B)}], \qquad
  \rho_{S,d}^{(B)} := \widetilde\Delta_S^{(B)}[\rho_S^{(B)}]
  = \sum_n \widetilde P_n^{(B)}\,\rho_S^{(B)}\,\widetilde P_n^{(B)}, \qquad
  \mathcal C_e^{(B)} := \mathcal S[\rho_{S,d}^{(B)}]-\mathcal S[\rho_S^{(B)}],
\end{equation}
with $\{\widetilde P_n^{(B)}\}_n$ a PVM in $\mathcal{H}^{(B)}_S$.
These quantities are related across descriptions as stated in the following theorem.

\begin{theorem}[Reduced-system entropy--coherence conservation law]
\label{thm:conservation_coherence_entanglement}
Let \(\rho_{SBE}^{(A)}\) be a state in the \(A\)-description, and let
\(\mathcal Q\) be the reduced QRF channel from the \(A\)-description to the
\(B\)-description. Let \(\{P_m^{(A)}\}_m\) and
\(\{\widetilde P_n^{(B)}\}_n\) be non-degenerate PVMs on
\(\mathcal H_S^{(A)}\) and \(\mathcal H_S^{(B)}\), respectively.

Assume that \(\mathcal Q\) is strongly block-preserving with respect to these
two PVMs, then the sum of the local entropy and the relative entropy of coherence of $S$ is frame-invariant:
\begin{equation}
\boxed{
\mathcal{E}_S^{(A)} + \mathcal{C}_e^{(A)}
=
\mathcal{E}_S^{(B)} + \mathcal{C}_e^{(B)}.
}
\label{eq:conservation_law}
\end{equation}
If the global state is pure, then \(\mathcal E_S^{(A)}\) and
\(\mathcal E_S^{(B)}\) are the entanglement entropies across the bipartitions
\(S:BE\) and \(S:AE\), respectively. In this case,
Eq.~\eqref{eq:conservation_law} is a conservation law for the sum of system
coherence and system entanglement entropy.
\end{theorem}

\begin{proof}
By definition of the relative entropy of coherence,
$\mathcal C_e(\rho_S)=\mathcal S(\rho_{S,d})-\mathcal S(\rho_S)$, hence
\begin{equation}
\mathcal E_S(\rho_S)+\mathcal C_e(\rho_S)
=\mathcal S(\rho_S)+\mathcal S(\rho_{S,d})-\mathcal S(\rho_S)
=\mathcal S(\rho_{S,d}).
\label{eq:Ec_equals_diag_entropy}
\end{equation}
Let $\rho_{SBE}^{(A)}$  be a state in the \(A\)-description and denote by
$\rho_S^{(B)}=\mathcal Q[\rho_{SBE}^{(A)}]$ the reduced system state in the
$B$-description. Under \eqref{eq:strong_block}, the
populations in the preserving basis are permuted:
\[
p_n^{(B)}=\Tr[\widetilde P_n^{(B)}\rho_S^{(B)}]
=\Tr[\mathcal Q^\dagger(\widetilde P_n^{(B)})\,\rho_{SBE}^{(A)}]
=\Tr[P_{\pi(n)}^{(A)}\rho_S^{(A)}]
=p_{\pi(n)}^{(A)}.
\]
Thus the diagonal spectra are related by a permutation, and since permutations leave the entropy unchanged,
we obtain
\(
\mathcal S(\rho_{S,d}^{(B)})
=\mathcal S(\rho_{S,d}^{(A)})
\).
Combining with~\eqref{eq:Ec_equals_diag_entropy} yields
\(
\mathcal E_S^{(A)}+\mathcal C_e^{(A)}=\mathcal E_S^{(B)}+\mathcal C_e^{(B)}
\).
\end{proof}
\begin{remark}
For degenerate block decompositions, strong block-preservation fixes the block
populations only up to permutation, and therefore conserves the block-population
entropy \(H(\{p_n\})=-\sum_n p_n\log p_n\). It does not, in general, conserve
the entropy of the block-diagonal state
\(\mathcal S(\sum_n P_n\rho_S P_n)\), which also depends on the spectra within
the degenerate blocks.
\end{remark}
Equation~\eqref{eq:conservation_law} can be read as a trade-off between two
system-level quantities associated with the chosen population--coherence
structure of \(S\): the local entropy of \(S\), quantified by
\(\mathcal E_S\), and the coherence of \(S\) with respect to the chosen PVM,
quantified by \(\mathcal C_e\). Under a strongly block-preserving QRF map,
these two quantities may individually change across descriptions, but their sum remains
fixed. Technically, the conserved quantity is the diagonal entropy
$\mathcal S(\rho_{S,\mathrm{d}})$, which depends only on the energy-population
statistics and is unaffected by its distribution between local coherence and correlations with the reference. When the global state $\rho^{(A)}_{SBE}$ is pure, the local entropy
\(\mathcal E_S\) is also an entanglement entropy; in that case, the same
identity becomes a conservation law for the sum of system coherence and
entanglement entropy. 

Note that \eqref{eq:conservation_law} is a purely kinematic statement that does not assume any particular reduced dynamics for the system. If, in addition, the dynamics in frame~$A$ is of pure-dephasing type, then the
block populations are constant in time, and therefore so is
$\mathcal S(\rho_{S,\mathrm d}(t))$. In that case, under the same hypotheses,
$\mathcal E_S(t)+\mathcal C_e(t)$ is not only frame-invariant but also conserved by the temporal evolution.

Ref.~\cite{Cepollaro} established conservation laws involving coherence and
entanglement of the \emph{reference} degrees of freedom under ideal QRF
transformations.\footnote{Note that while the conservation law does not hold for non-ideal QRFs in general \cite{Cepollaro}, Ref.~\cite{delaHamette_entropy_2026} provides a bound on the entropy difference for non-ideal frames.}  In contrast, Theorem~\ref{thm:conservation_coherence_entanglement} is formulated at the reduced-system level: it concerns the local entropy and coherence of \(S\), and follows from strong block-preservation of the reduced QRF channel. Moreover, the statement is formulated directly in the open-system setting with an environment \(E\), so that the entanglement interpretation, when the global state is pure, refers to the bipartitions \(S:BE\) in the \(A\)-description and \(S:AE\) in the \(B\)-description. Another difference is that, in Ref.~\cite{Cepollaro}, the dephased state $\rho_d$, used to define the relative entropy of coherence, is always obtained by pinching in the basis induced by the symmetry group. Here, we show that the tradeoff between entropy and coherence holds for any system PVM that satisfies  the strong block-preserving condition. Finally, although the general framework of this paper is developed for ideal QRFs, the proof of the theorem does not use the explicit \(L^2(G)\) structure of the reference, but only the channel condition in Eq.~\eqref{eq:strong_block}.  
The conservation law
\eqref{eq:conservation_law} therefore provides a frame-invariant system-level quantity that is preserved under a broad class of QRF channels.

\subsection{Classical vs quantum signatures of reduced QRF transformations}
\label{subsec:Quantumvsclassical}
The reference-frame changes considered in this work are changes between
quantum systems, namely QRFs. This is manifest at the level of the total
description, which includes the open system, the environment, and the reference
systems themselves. However, once one restricts attention to the reduced
description of the open system alone, the induced action of a QRF transformation
need not display quantum features. In this section we identify a
\emph{classical misalignment} regime, in which the reduced action on \(S\)
admits an interpretation as random uncertainty about the frame configuration,
and contrast it with quantum reduced-frame effects. Distinguishing these regimes is operationally important: it tells us when frame
dependence at the level of \(S\) can be understood as
classical uncertainty about the reference configuration, and when this
classical explanation fails.

\begin{definition}[Classical misalignment representation]
\label{def:classical_misalignment}
We say that the action of a reduced QRF channel $\mathcal{Q}$ on a state $\rho^{(A)}_{SBE}$ admits a
\emph{classical misalignment representation} with respect to the group action
\(\{V_S(g)\}_{g\in G}\) if its induced action on \(S\) can be written as an
incoherent average over frame orientations,
\begin{equation}
\rho_S^{(B)}=\mathcal{Q}[\rho^{(A)}_{SBE}]= \int_G dg\; p(g)\, V^{\dagger}_S(g)\, \rho^{(A)}_{S} \, V_S(g),
\qquad p(g)\ge 0,\quad \int_G dg\,p(g)=1,
\label{eq:classical_misalignment}
\end{equation}
where $dg$ is the (left) Haar measure on $G$ and $V_S(g)$ is a unitary
representation on $\mathcal H_S$. 
\end{definition}

This corresponds to classical uncertainty
about the parameter $g$. Perfect knowledge reduces to a single
unitary $V_S(g_0)$ (i.e.\ $p(g)=\delta(g-g_0)$), otherwise the incoherent average over the different orientations represents an uncertainty in the alignment of the new reference frame relative to the old -- a \emph{misalignment}. Any classical misalignment channel is random-unitary and, in particular, unital:
$\Phi_{\mathrm{cl}}(\mathbb I)=\mathbb I$.\\

For ideal QRF transformations, a classical misalignment representation is
obtained whenever, at the time of the QRF transformation, the state factorizes
as \(\rho_{SBE}^{(A)}=\rho_{SE}^{(A)}\otimes\rho_B^{(A)}\), and the group action
factorizes as \(V_{SE}(g)=V_S(g)\otimes V_E(g)\).
Under these assumptions, the perspective-change operator reads
\[
S_{A\to B}=\int_G dg\, |g^{-1}\rangle_A\!\langle g|_B \otimes V^{\dagger}_S(g)\otimes V^{\dagger}_E(g),
\]
and the reduced state in frame $B$ becomes
\begin{align}
\rho_S^{(B)}
&= \Tr_{AE}\!\big[S_{A\to B}\,(\rho^{(A)}_{SE}\otimes\rho^{(A)}_B)\,S_{A\to B}^\dagger\big] \nonumber\\
&= \iint_G dg\,dg'\;\underbrace{\langle g|\rho^{(A)}_B|g'\rangle}_{=:k(g,g')}\ \langle g'^{-1}|g^{-1}\rangle_A
\,\Tr_E\!\Big[(V^{\dagger}_S(g)\!\otimes\!V^{\dagger}_E(g))\,\rho^{(A)}_{SE}\,(V_S(g')\!\otimes\!V_E(g'))\Big]
\,.
\label{eq:kernel_form}
\end{align}
In the ideal limit, $\langle g|g'\rangle=\delta(g-g')$ and thus
$\langle g^{-1}|g'^{-1}\rangle=\delta(g-g')$, only the diagonal kernel
$k(g,g)$ contributes. Hence, defining the normalized probability density
$p(g):=k(g,g)$
, we obtain
\begin{align}
\rho_S^{(B)}
&= \int_G dg\, p(g)\,\Tr_E\!\Big[(V^{\dagger}_S(g)\!\otimes\!V^{\dagger}_E(g))\,\rho^{(A)}_{SE}\,(V_S(g)\!\otimes\!V_E(g))\Big] \nonumber\\
&= \int_G dg\, p(g)\,V^{\dagger}_S(g)\,\rho_S^{(A)}\,V_S(g),
\label{eq:random_unitary}
\end{align}
where in the last step we used invariance of the partial trace under local
unitaries on $E$.
Equation \eqref{eq:random_unitary} provides a concrete instance in which a
quantum reference frame, once discarded, induces a reduced transformation
that nevertheless lies in the classical misalignment class: the effective
distribution $p(g)=\langle g|\rho^{(A)}_B|g\rangle$ depends only on the diagonal
reference statistics in the orientation basis.
On the other hand, considering the more general situations of non-ideal reference frames or non-factorized group actions allows us to leave the classical misalignment class of reduced QRF transformations.
In particular, non-ideal reference frames have
non-orthogonal orientation states, so that off-diagonal contributions \(g\neq g'\)
are not removed by the trace over the old reference: 
in this case, Eq.(\ref{eq:random_unitary}) has to be replaced
with a double-kernel structure (compare with Eq.(\ref{eq:kernel_form})),
\begin{equation}
\rho^{(B)}_S
=
\iint_G dg\,dg'\,
f(g,g')\,
\Tr_E\!\left[
\bigl(V^{\dagger}_S(g)\otimes V^{\dagger}_E(g)\bigr)
\rho_{SE}^{(A)}
\bigl(V_S(g')\otimes V_E(g')\bigr)
\right],
\label{eq:double_kernel_general}
\end{equation}
with $f(g, g') = k(g,g')\langle g'^{-1}|g^{-1}\rangle_A$, where the off-diagonal contributions are controlled 
by reference coherences, \(\langle g|\rho_B^{(A)}|g'\rangle\) with \(g\neq g'\), together with
the non-orthogonality of the orientation states.
This suggests that QRF imperfections need not only be viewed as errors: in some
settings they can enlarge the class of effective reduced transformations beyond
classical misalignment.
Similarly, if, instead, the group action $V_{SE}(g)$ does not factorize\footnote{Non-factorizing actions are usually excluded in idealized QRF models by choosing a subsystem split for which the group acts as a product. In general, factorization depends on how one partitions the Hilbert space into accessible and inaccessible degrees of freedom; e.g.\ Lorentz boosts can couple spin and momentum, yielding a non-product representation for the spin--momentum split~\cite{Giacomini2018}.},
tracing out $E$ we are left with
\begin{equation}
\rho_S^{(B)}
=\int_G dg\,p(g)\;\Tr_E\!\big[V^{\dagger}_{SE}(g)\,\rho_{SE}^{(A)}\,V_{SE}(g)\big],
\label{eq:mixture_general_channels}
\end{equation}
where the contribution associated with each $g$ depends, in general, on the full
joint state $\rho_{SE}^{(A)}$ and need not reduce to a unitary conjugation on
$\rho_S^{(A)}$ alone. Finally, if we instead relax the assumption that the state factorizes but assume a general pure state with Schmidt-decomposition $\rho_{SBE}^{(A)}=\sum_{j,j'}\sqrt{\lambda_j\lambda_{j'}}|a_j\rangle\langle a_{j'}|_{SE}\otimes |b_j\rangle \langle b_{j'}|_B$, we obtain
\begin{equation}
    \rho_S^{(B)}=\sum_{j,j'}\sqrt{\lambda_j\lambda_{j'}} \int_G dg\, \underset{=:k_{j,j'}(g)}{\underbrace{\langle g|b_j\rangle\langle b_{j'}|g\rangle}} \;\Tr_E\!\big[V^{\dagger}_{SE}(g)|a_j\rangle\langle a_{j'}|V_{SE}(g) \big].
\end{equation}
In general, this can no longer be understood in terms of a classical mixture
of frame orientations due to entanglement across the \(B-SE\) bipartition.

Let us stress that the classical-misalignment representation introduced here and the
block-preserving structure developed above are, in general, independent
properties of a reduced QRF transformation. A reduced QRF channel may be
block-preserving without admitting a random-unitary representation on \(S\),
for example when correlations with inaccessible degrees of freedom affect the
reduced action while preserving the chosen block populations. Conversely, a
classical-misalignment representation need not be block-preserving with respect
to a prescribed PVM: if the unitaries \(V_S(g)\) rotate the chosen blocks, the
random-unitary average can mix populations and coherences in that decomposition.

\subsubsection*{A witness based on unitality}
The previous discussion concerns the reduced transformation induced by the QRF
change alone. In the open-system setting considered in this work, however, the
state to which the reduced QRF channel is applied is itself generated by a
global evolution. Thus, at each time \(t\), the accessible state in frame \(B\)
is
\begin{equation}
\rho_S^{(B)}(t)
=
\mathcal Q\!\left[\rho_{SBE}^{(A)}(t)\right].
\end{equation}
We say that the combined action of the evolution and the QRF change admits a
classical misalignment form at time \(t\) if
\begin{equation}
\rho_S^{(B)}(t)
=
\mathcal G_{p_t}\!\left(\rho_S^{(A)}(t)\right)
:= \int_G dg\, p_t(g)\,V^{\dagger}_S(g)\,\rho_S^{(A)}(t)\,V_S(g),
\label{eq:classical_misalignment_dynamical}
\end{equation}
for some probability density \(p_t(g)\) on \(G\), 
$p_t(g)\ge 0$ and $\int_G dg\,p_t(g)=1$, and where indeed
\begin{equation}
\rho_S^{(A)}(t) =\Tr_{BE}\left[\rho_{SBE}^{(A)}(t)\right] = \Tr_{BE}\!\Big[U_{SBE}(t)\rho^{(A)}_{SBE}(0)U_{SBE}^\dagger(t)\Big].
\end{equation}
If, in addition, the reduced dynamics in frame \(A\) is described by a channel
\(\Phi_t\), so that
\(\rho_S^{(A)}(t)=\Phi_t[\rho_S^{(A)}(0)]\), the corresponding reduced dynamics
in frame \(B\) takes the form
\[
\Lambda_t
=
\mathcal G_{p_t}\circ \Phi_t .
\]
This leads to the following unitality witness.

\begin{proposition}[Unitality preserved under classical misalignment]
\label{thm:nogo_unital}
Let
\(\{\Phi_t\}_{t\ge 0}\) be a family of CPTP maps on \(S\) that is unital for all
times,
\begin{equation}
\Phi_t(\mathbb I_S)=\mathbb I_S \qquad \forall\, t\ge 0 .
\label{eq:unital_assumption}
\end{equation}
Let \(\{\mathcal G_{p_t}\}_{t\ge 0}\) be a 
classical misalignment channel,
\begin{equation}
\mathcal G_{p_t}(\rho)
=
\int_G dg\,p_t(g)\,
V^{\dagger}_S(g)\,\rho\,V_S(g),
\qquad
p_t(g)\ge 0,\quad
\int_G dg\,p_t(g)=1 .
\label{eq:twirl_map_t}
\end{equation}
Then
\[
\Lambda_t=\mathcal G_{p_t}\circ \Phi_t
\]
is unital for all \(t\ge 0\).
\end{proposition}
Operationally, Proposition~\ref{thm:nogo_unital} implies that classical misalignment cannot turn unital dynamics into non-unital dynamics: it can only
reshape unital noise on \(S\). In particular, it cannot generate non-unital
features such as relaxation towards a preferred state or dissipative population
bias when the reduced dynamics in frame \(A\) is unital. Equivalently, if the
dynamics in frame \(A\) is known to be unital while the reduced dynamics
inferred in frame \(B\) satisfies
\(\Lambda_t(\mathbb I_S)\neq\mathbb I_S\), then the frame dependence cannot be
accounted for by a classical misalignment channel of the form
\eqref{eq:twirl_map_t}. Such non-unitality therefore provides an operational witness of reduced QRF effects beyond classical random frame misalignment.

\section{Pure dephasing: a complete criterion for population preservation}

\label{sec:gpreserving_pd}
In this section, we consider the reduced-channel framework in the special case of pure dephasing, the paradigmatic noise model in which populations are preserved while coherences decay. 
Our goal is to determine under which conditions the pure-dephasing dynamics is preserved after a change of QRF, for all factorized initial states in the $S-BE$ partition.
We show that this happens precisely when the reduced QRF channel satisfies a \emph{dynamical compatibility} condition with the controlled evolution generating pure dephasing in the original frame.

As discussed above, block-preservation is a \emph{kinematic} condition on the reduced QRF channel: it prevents population--coherence mixing with respect to a chosen block structure. By itself, however, it does not guarantee that a dynamics that preserves populations in frame $A$ remains population-preserving in frame $B$. After the change of frame and the trace over inaccessible degrees of freedom, the induced coarse-graining on $S$ may still generate an effective redistribution of populations among the blocks. The missing ingredient is therefore not additional kinematic structure, but compatibility between the reduced coarse-graining induced by $\mathcal Q$ and the conditional dynamics underlying pure dephasing in frame $A$.\\

Recall from Sec.~\ref{subsec:pure_dephasing} that pure dephasing in the
$A$-description is generated by a controlled joint propagator of the form
\begin{equation}
U^{(A)}_{SBE}(t)=\sum_m P_m^{(A)}\otimes U_m^{(A)}(t),
\label{eq:controlled_unitary_repeat}
\end{equation}
with respect to the energy projectors $\{P_m^{(A)}\}_m$, where each
$U_m^{(A)}(t)$ acts on the joint inaccessible degrees of freedom $BE$. This structure guarantees that the populations associated with $\{P_m^{(A)}\}_m$ are
constant.

\begin{definition}[Controlled-dynamics compatible reduced QRF channel]
\label{def:G_preserving}
Assume that the reduced QRF channel $\mathcal Q$ is block-preserving with
respect to the energy projectors $\{P_m^{(A)}\}_m$ and a fixed PVM
$\{\widetilde P_n^{(B)}\}_n$ on $\mathcal H_S^{(B)}$, so that
\begin{equation}
\mathcal Q^\dagger(\widetilde P_n^{(B)})
=
\sum_m P_m^{(A)}\otimes \Pi_{nm}^{(A)},
\label{eq:block_preserving_def_only_dyn}
\end{equation}
with $\Pi_{nm}^{(A)}\in\mathcal B(\mathcal H_{BE}^{(A)})$.
Assume moreover that the joint propagator
in the $A$-description has the controlled form
\eqref{eq:controlled_unitary_repeat}. We say that $\mathcal Q$ is  \emph{dynamically compatible} (with respect to this
block structure and controlled dynamics) if, for all $m,n$ and all $t\ge 0$,
\begin{equation}
\big[\Pi_{nm}^{(A)},\,U_m^{(A)}(t)\big]=0.
\label{eq:G_preserving_commutation}
\end{equation}
\end{definition}

Dynamical compatibility is, in general, a strengthening of block-preservation,
obtained by imposing the commutation condition
\eqref{eq:G_preserving_commutation} on the operators $\Pi_{nm}^{(A)}$
that appear in \eqref{eq:block_preserving_def_only}. For strongly block-preserving reduced QRF channels,
however, this extra requirement follows automatically. Indeed, if $\mathcal Q^\dagger\!\big[\widetilde P_n^{(B)}\big]= P_{\pi(n)}^{(A)}\otimes \mathbb I_{BE}$
for all $n$, then $\Pi_{nm}^{(A)}=\delta_{m,\pi(n)}\,\mathbb I_{BE}$,
and hence $[\Pi_{nm}^{(A)},U_m^{(A)}(t)]=0$
for all  $n,m,\ \forall t\ge 0$ holds trivially. Thus, once the controlled dynamics
\eqref{eq:controlled_unitary_repeat} is fixed, strong block-preservation already
implies dynamical compatibility. This holds in particular in the factorized regime
\(V_{SE}(g)=V_S(g)\otimes V_E(g)\), when
\([H_S^{(A)},V_S(g)]=0\) for all \(g\in G\). In this case
\(V_S(g)\) commutes with the spectral projectors \(P_m^{(A)}\), so the reduced
QRF channel is strongly block-preserving with respect to the corresponding
energy PVMs in the two frame descriptions, and is therefore dynamically
compatible with any controlled evolution of the form
\eqref{eq:controlled_unitary_repeat}.

Having introduced the criterion of dynamical compatibility, we are now in a position to formalize when a change of QRF 
preserves a fixed population decomposition, uniformly over all factorized initial states.
\begin{theorem}[Dynamical compatibility $\Leftrightarrow$ preservation of populations]
\label{thm:gpreserving_preservation}
Assume that the joint propagator in the $A$-description has the controlled form
\begin{equation}
U^{(A)}(t)=\sum_m P_m^{(A)}\otimes U_m^{(A)}(t),
\label{eq:controlled_unitary_repeat_thm}
\end{equation}
with respect to a non-degenerate energy PVM $\{P_m^{(A)}\}_m$ on
$\mathcal H_S^{(A)}$.
Let $\mathcal Q$ be the reduced QRF channel of Definition
\ref{def:reduced_qrf_channel}, and assume that $\mathcal Q$ is
block-preserving with respect to $\{P_m^{(A)}\}_m$ and to a fixed PVM
$\{\widetilde P_n^{(B)}\}_n$ on $\mathcal H_S^{(B)}$, i.e.
\begin{equation}
\mathcal Q^\dagger(\widetilde P_n^{(B)})
=
\sum_m P_m^{(A)}\otimes \Pi_{nm}^{(A)} .
\label{eq:block_preserving_def_only_thm}
\end{equation}

Then the following are equivalent:
\begin{enumerate}[label=(\roman*)]
\item $\mathcal Q$ is dynamically compatible, i.e.
\begin{equation}
[\Pi_{nm}^{(A)},\,U_m^{(A)}(t)]=0
\qquad
\forall n,m,\ \forall t\ge 0.
\label{eq:G_preserving_commutation_thm}
\end{equation}

\item The population structure is preserved in the
$B$-description: for every factorized initial state $\rho_{SBE}^{(A)}(0)=\rho_S^{(A)}(0)\otimes \rho_{BE}^{(A)}(0)$,
the populations associated with the \emph{same fixed} PVM
$\{\widetilde P_n^{(B)}\}_n$ are constant,
\begin{equation}
p_n^{(B)}(t):=
\Tr\!\big[\widetilde P_n^{(B)}\,\rho_S^{(B)}(t)\big]
=
\Tr\!\big[\widetilde P_n^{(B)}\,\rho_S^{(B)}(0)\big]
\qquad
\forall n,\ \forall t\ge 0.
\label{eq:pd_preservation_pop_thm}
\end{equation}\end{enumerate}
\end{theorem}
\begin{proof}
Consider an arbitrary
factorized initial state,
the joint state at time $t$ is
\[
\rho_{SBE}^{(A)}(t)
=
\sum_{m,m'}
P_m^{(A)}\,\rho_S^{(A)}(0)\,P_{m'}^{(A)}
\otimes
U_m^{(A)}(t)\,\rho_{BE}^{(A)}(0)\,U_{m'}^{(A)\dagger}(t).
\]
Using \eqref{eq:block_preserving_def_only_thm}, the population of the $n$-th
block in the $B$-description is
\begin{align}
p_n^{(B)}(t)
&=
\Tr\!\big[\widetilde P_n^{(B)}\,\rho_S^{(B)}(t)\big]
=
\Tr\!\big[\mathcal Q^\dagger(\widetilde P_n^{(B)})\,\rho_{SBE}^{(A)}(t)\big]
\nonumber\\
&=
\sum_m
\Tr\!\big[P_m^{(A)}\rho_S^{(A)}(0)\big]\,
\Tr\!\Big[
\Pi_{nm}^{(A)}\,
U_m^{(A)}(t)\rho_{BE}^{(A)}(0)U_m^{(A)\dagger}(t)
\Big],
\label{eq:pnB_general_formula_thm}
\end{align}
where only the diagonal terms $m=m'$ survive because
$P_m^{(A)}P_{m'}^{(A)}=\delta_{mm'}P_m^{(A)}$.

\emph{(i)$\Rightarrow$(ii).}
If $\mathcal Q$ is dynamically compatible, then
$[\Pi_{nm}^{(A)},U_m^{(A)}(t)]=0$ for all $m,n,t$, hence
\[
\Tr\!\Big[
\Pi_{nm}^{(A)}\,
U_m^{(A)}(t)\rho_{BE}^{(A)}(0)U_m^{(A)\dagger}(t)
\Big]
=
\Tr\!\big[\Pi_{nm}^{(A)}\,\rho_{BE}^{(A)}(0)\big].
\]
Therefore each term in \eqref{eq:pnB_general_formula_thm} is independent of
$t$, and so $p_n^{(B)}(t)$ is constant for all $n$, for all $t\ge 0$, and for
all factorized initial states.

\emph{(ii)$\Rightarrow$(i).}
Assume now that \eqref{eq:pd_preservation_pop_thm} holds for every factorized
initial state. Fix $m$ and choose
$\rho_S^{(A)}(0)=P_m^{(A)}$, so that \eqref{eq:pnB_general_formula_thm} reduces
to
\[
p_n^{(B)}(t)
=
\Tr\!\Big[
\Pi_{nm}^{(A)}\,
U_m^{(A)}(t)\rho_{BE}^{(A)}(0)U_m^{(A)\dagger}(t)
\Big].
\]
By assumption, this quantity is independent of $t$ for every choice of
$\rho_{BE}^{(A)}(0)$. Since this holds for all states on $BE$, it follows that
\[
U_m^{(A)\dagger}(t)\,\Pi_{nm}^{(A)}\,U_m^{(A)}(t)=\Pi_{nm}^{(A)}
\qquad
\forall n,m,\ \forall t\ge 0,
\]
which is equivalent to \eqref{eq:G_preserving_commutation_thm}. Hence
$\mathcal Q$ is dynamically compatible.
\end{proof}
\begin{remark}[Time-dependent effects]
If the reduced QRF transformation or the preserving PVM is time-dependent, the
effects \(\Pi_{nm}^{(A)}\) become time-dependent. In that case the appropriate
condition for preservation of populations is not simply
\([\Pi_{nm}^{(A)}(t),U_m^{(A)}(t)]=0\), but rather
\[
U_m^{(A)\dagger}(t)\Pi_{nm}^{(A)}(t)U_m^{(A)}(t)
=
\Pi_{nm}^{(A)}(0).
\]
The commutation condition used above is recovered when
\(\Pi_{nm}^{(A)}(t)\) is time independent.
\end{remark}

Theorem~\ref{thm:gpreserving_preservation} shows that dynamical compatibility is
equivalent to the constancy of the populations $p_n^{(B)}(t)$ for \emph{all}
factorized initial states
$\rho_{SBE}^{(A)}(0)=\rho_S^{(A)}(0)\otimes\rho_{BE}^{(A)}(0)$. If one instead fixes a single preparation $\rho_{BE}^{(A)}(0)$ of the inaccessible
degrees of freedom, then dynamical compatibility remains sufficient, but in general no
longer necessary: population constancy may occur accidentally for that specific
state due to cancellations in the expectation values 
$\Tr[\Pi_{nm}^{(A)}U_m^{(A)}(t)\rho_{BE}^{(A)}(0)U_m^{(A)\dagger}(t)]$.
Such state-dependent preservation is not robust and does not characterize the
transformation. Moreover, if the initial state is correlated,
$\rho_{SBE}^{(A)}(0)\neq \rho_S^{(A)}(0)\otimes\rho_{BE}^{(A)}(0)$, the reduced
evolution on $S$ need not define a CPTP map for arbitrary system preparations.
Nevertheless, assuming the controlled form
\eqref{eq:controlled_unitary_repeat}, the dynamical compatibility condition
\eqref{eq:G_preserving_commutation} remains sufficient to ensure that the
populations associated with the fixed projectors $\{\widetilde P_n^{(B)}\}$ are
constant along the resulting trajectory in the $B$-description.\\

\begin{example}[Dynamical compatibility without strong block-preservation] \label{ex:z2_dynamic_compatibility}
We continue Example~\ref{ex:z2_block_preserving_not_strong}, where the reduced QRF channel \(\mathcal Q\) was shown to be block-preserving with respect to the input energy PVM \(\{P_m^{(A)}\}_{m=0,1}\) and the output population PVM \(\{\widetilde P_n^{(B)}\}_{n=0,1}\), but not strongly block-preserving. The corresponding effects, given in Eq.~\eqref{eq:z2_block_preserving_effects}, are diagonal in the pointer basis \(\{\ket0_B,\ket1_B\}\) of the reference \(B\) and can be written as \[ \Pi_{nm}^{(A)} = \ket{b_{nm}}\!\bra{b_{nm}}_B\otimes\mathbb I_E, \] with \(b_{00}=b_{11}=0\) and \(b_{01}=b_{10}=1\).
Assume now that the conditional unitaries generating the pure-dephasing dynamics in frame \(A\) preserve the pointer decomposition of \(B\), namely \begin{equation} U_m^{(A)}(t) = \sum_{b=0,1} \ket b\!\bra b_B\otimes W_{m,b}^{(E)}(t), \label{eq:z2_dynamic_compatible_unitary} \end{equation} where \(W_{m,b}^{(E)}(t)\) are unitaries on \(E\). Then \[ \Pi_{nm}^{(A)}U_m^{(A)}(t) = \ket{b_{nm}}\!\bra{b_{nm}}_B \otimes W_{m,b_{nm}}^{(E)}(t) = U_m^{(A)}(t)\Pi_{nm}^{(A)} . \] The equality holds because \(\Pi_{nm}^{(A)}\) selects one pointer block of \(B\) and acts as the identity on \(E\) within that block. Therefore \begin{equation} \bigl[\Pi_{nm}^{(A)},U_m^{(A)}(t)\bigr]=0 \qquad \forall n,m,\ \forall t\geq0 . \label{eq:z2_dynamic_compatibility_condition} \end{equation} 
Hence the reduced QRF channel is dynamically compatible with the pure dephasing
evolution. By Theorem~\ref{thm:gpreserving_preservation}, for every factorized
initial state
$\rho_{SBE}^{(A)}(0)=\rho_S^{(A)}(0)\otimes\rho_{BE}^{(A)}(0)$,
the populations associated with the output PVM
$\{\widetilde P_n^{(B)}\}_{n=0,1}$ are constant in time,
\begin{equation}
\Tr\!\left[ \widetilde P_n^{(B)} \rho_S^{(B)}(t) \right]
=
\Tr\!\left[ \widetilde P_n^{(B)} \rho_S^{(B)}(0) \right]
\qquad
\forall\,n,\quad \forall\,t\geq0.
\end{equation}

This example shows that strong block-preservation is not necessary for dynamical
compatibility: the reduced QRF channel considered here is block-preserving but
not strongly block-preserving, while the conditional evolution satisfies the
dynamical compatibility condition~\eqref{eq:z2_dynamic_compatibility_condition}.

\end{example}

Theorem~\ref{thm:gpreserving_preservation} concerns the preservation of a fixed
population decomposition and does not by itself control the evolution of the
off-diagonal terms in the transformed frame. Preservation of pure dephasing in
the strict sense --- namely, constancy of populations together with
multiplicative evolution of each coherence --- requires additional structure. 
In the next section, we present such a scenario: the
classical-misalignment regime of Sec.~\ref{subsec:Quantumvsclassical}, where
the induced reduced action on \(S\) is random-unitary.

\section{Operational meaning of frame-induced decoherence}
\label{sec:Ramsey}
The previous sections identified conditions under which a change of QRF
preserves the population structure of a pure-dephasing dynamics. Here we ask
what happens to the corresponding coherences when population preservation holds.
In the regime considered below, the transformed dynamics retains the
pure-dephasing form, but the coherences acquire an additional frame-dependent
factor. As a result, the decoherence rate inferred in the transformed frame
splits into an environment-induced contribution and a reference-induced one.

Ramsey interferometry gives this decomposition direct operational meaning and, in particular, clarifies the distinction between describing the same physical interferometer from different quantum reference frames and performing local Ramsey experiments tied to different phase standards.
This type of experimental setup may also indicate a way to test the validity of the QRF perspective~\cite{Overstreet:2022zgq}, using the fact that, in the frame of the detectors (hence fixing their configuration), measurements only testing the relational degrees of freedom between two external particles should give the same outcome.

\subsection{Frame factors and additive rate decomposition}
\label{subsec:frame_factors}
We now specify the assumptions under which the reduced QRF transformation acts
multiplicatively on coherences. We specialize to the classical-misalignment
regime introduced in Sec.~\ref{subsec:Quantumvsclassical}. As we derived there, if (i) in the QRF of $A$ the state of $B$ is factorized with respect to the state of $S$ and $E$, (ii) there is no interaction between $B$ and $SE$, and (iii) the group action factorizes, then the reduced description of $S$ in frame $B$ takes the random-unitary form  of Eq.~\eqref{eq:classical_misalignment_dynamical},
which we recall here for convenience:
\begin{equation}
\rho_S^{(B)}(t)
=
\int_G dg\,p_t(g)\,V^{\dagger}_S(g)\,\rho_S^{(A)}(t)\,V_S(g),
\label{eq:random_unitary_general}
\end{equation}
with $p_t(g)\ge 0$ and $\int_G dg\,p_t(g)=1$ and $V_S(g)$ being the group action on system $S$.
Taking matrix elements in the energy basis $\{\ket n\}$ of $H_S^{(A)}$, which we assume to be non-degenerate, gives
\begin{equation}
\rho_{nm}^{(B)}(t)
=
\sum_{n',m'}
\left[
\int_G dg\,p_t(g)\,V^*_{n'n}(g)V_{m'm}(g)
\right]
\rho_{n'm'}^{(A)}(t),
\qquad
V_{nn'}(g):=\langle n|V_S(g)|n'\rangle.
\label{eq:matrix_element_mixing}
\end{equation}
In general, Eq.~\eqref{eq:matrix_element_mixing} shows that a reduced QRF transformation may mix matrix elements, 
unless the group action preserves the energy decomposition of $S$. We will
refer to the latter situation as the \emph{Hamiltonian-symmetric} regime: the representation $V_S(g)$
is diagonal in the energy basis,
$V_S(g)\ket n = e^{-i\phi_n(g)}\ket n$. In this case, the QRF contribution does not transfer population between energy
eigenspaces, and each energy coherence is simply rescaled by a frame-dependent
factor,
\begin{equation}
\label{eq:Fnm_phase_case}
\rho_{nm}^{(B)}(t)=F_{nm}(t)\,\rho_{nm}^{(A)}(t),
\qquad
F_{nm}(t):=\int_G dg\,p_t(g)\,e^{i(\phi_n(g)-\phi_m(g))}.
\end{equation}
We refer to $F_{nm}(t)$ as a \emph{frame factor}, since it depends only on the reference-frame degrees of freedom through the distribution $p_t(g)$ and it quantifies the contribution of the QRF change alone to the $(n,m)$ coherence. In operational terms, $p_t(g)$ is the sampling profile induced by the reference state; hence variations of
$F_{nm}(t)$ directly reflect how the reference-frame degrees of freedom affect coherence in the
transformed description.

Recalling the general definition of the (time-local) decoherence rate introduced in Sec.~\ref{subsec:pure_dephasing}, Eq.~\eqref{eq:decoherence_rate}, and applying it to the coherences in frames $A$ and $B$, the relation
\eqref{eq:Fnm_phase_case} immediately implies the additive decomposition
\begin{equation}
\gamma_{nm}^{(B)}(t)
=
\underbrace{\gamma_{nm}^{(A)}(t)}_{\substack{\textit{environment-induced}\\\textit{dephasing}}}
+
\underbrace{\gamma_{nm}^{\mathrm{fr}}(t)}_{\textit{frame-induced contribution}},
\qquad
\gamma_{nm}^{\mathrm{fr}}(t):=-\frac{d}{dt}\ln|F_{nm}(t)|.
\label{eq:rate_decomposition}
\end{equation}
Intuitively, a rigid translation of $p_t(g)$ corresponds to a mere drift of the reference (a re-labeling
of the group element $g$) with no change in its ``sharpness'': the reference remains equally good and
can only shift the phase of the coherences, leaving $|F_{nm}(t)|$ unchanged (see Appendix \ref{sec:rigidtranslations}). By contrast, non-rigid behavior means that the sampling
profile is reshaped in time---e.g.\ it broadens or is distorted---signaling a change in the quality of
the reference as a symmetry-breaking resource. This may arise due to noise-induced broadening or phase-space shearing under nonlinear Hamiltonians. In this case, such reshaping
changes the relevant Fourier component of $p_t(g)$, so that the modulus $|F_{nm}(t)|$ can vary and the
QRF change alone contributes to the effective decoherence rate, $\gamma_{nm}^{\mathrm{fr}}(t)\neq 0$. In Appendix \ref{app:degenerate_blocks}, we further discuss the extension to degenerate energy levels and explicitly time-dependent QRF transformations.

\subsection{Ramsey interferometry as an operational probe}
\label{subsec:ramsey_protocol}

Ramsey interferometry provides a direct operational probe of the decoherence rate in Eq.~\eqref{eq:decoherence_rate} \cite{Ramsey1950,WallsMilburn2008}. For any two-level subspace $\mathrm{span}\{\ket n,\ket m\}$ of energy eigenstates, a Ramsey sequence prepares a superposition, lets it evolve freely for a waiting time $t$, and maps the resulting coherence back to a population imbalance. Technical details of the protocol are collected in Appendix~\ref{ap:ramasey protocol}. In the effective qubit picture, the Ramsey fringe visibility is
\begin{equation}
\mathcal V_{nm}(t)=2|\rho_{nm}(t)|,
\label{eq:ramsey_visibility}
\end{equation}
so that the time-local decoherence rate is directly obtained from the decay of Ramsey contrast,
\begin{equation}
\gamma_{nm}(t)
=
-\frac{d}{dt}\ln |\rho_{nm}(t)|
=
-\frac{d}{dt}\ln |\mathcal V_{nm}(t)|.
\label{eq:ramsey-rate}
\end{equation}
Ramsey interferometry therefore gives the decoherence rate a direct operational meaning: it is the time-local decay rate of fringe visibility. When the frame factor modifies this visibility in a controlled and measurable way, as in Eq.~\eqref{eq:Fnm_phase_case}, Ramsey interferometry is a natural diagnostic of frame-induced decoherence. 

We now distinguish between two conceptually different scenarios: one in which
the same physical experiment is described from different QRFs, and one in which
local experiments are performed relative to different phase references.
In the first scenario, different observers, Alice and Bob, describe the \emph{same physical Ramsey interferometer} from different QRF perspectives. Assume that Alice, with reference frame $A$, implements a Ramsey experiment on $S$. Let \(M_\pm^{(A)}(t)\) denote the measurement effects associated with the two
Ramsey outcomes in Alice's description, written in the Heisenberg picture so as
to absorb the two \(\pi/2\) pulses, the free evolution, and the final projective
measurement.
Then Bob, describing that same experiment with respect to reference frame $B$, must transform both the state and the measurement under the QRF change,
\begin{equation}
\rho_{SAE}^{(B)}(0)
=
\hat S_{A\to B}\,\rho_{SBE}^{(A)}(0)\,\hat S_{A\to B}^\dagger,
\qquad
\tilde{M_\pm}^{(B)}(t)
=
\hat S_{A\to B}\,M_\pm^{(A)}(t)\,\hat S_{A\to B}^\dagger.
\label{eq:same_interferometer_transformation}
\end{equation}
This is the same operational point emphasized in Ref.~\cite{Giacomini2018}: changing perspective for a given experiment requires transforming both preparation and measurement. By unitarity of $\hat S_{A\to B}$ and cyclicity of the trace, the outcome statistics are unchanged. Hence the Ramsey visibility, and therefore the operational decoherence rate associated with that interferometer, are frame-invariant. Note, however, that the measurement $M_\pm^{(B)}(t)$ no longer corresponds to a local measurement on $S$ but rather a joint operation on both $S$ and $A$ as illustrated in Fig.~\ref{fig:Ramsey1}.

\begin{figure}[H]
    \centering
    \begin{subfigure}{0.4\linewidth}
        \centering
          \includegraphics[width=\linewidth,trim=0 0 0 0,clip]{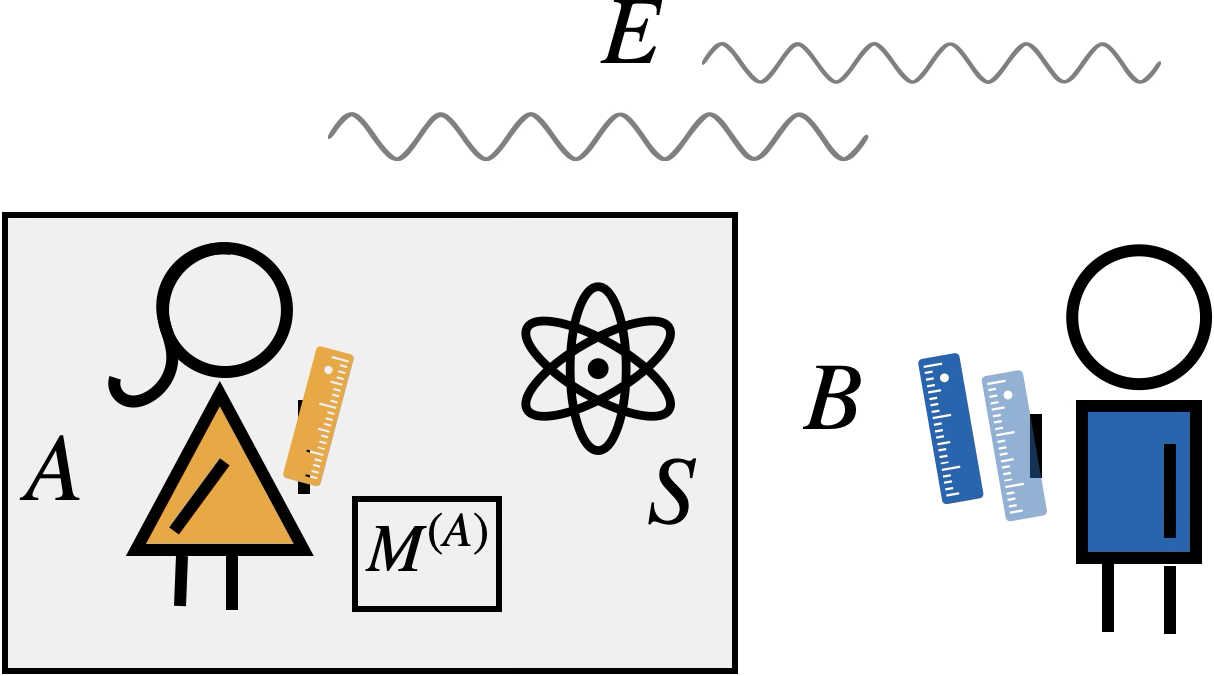}
          \caption{Alice's Perspective}
    \end{subfigure}\hspace{1.5cm}
   \begin{subfigure}{0.4\linewidth}
        \centering
          \includegraphics[width=\linewidth,trim=0 0 0 0,clip]{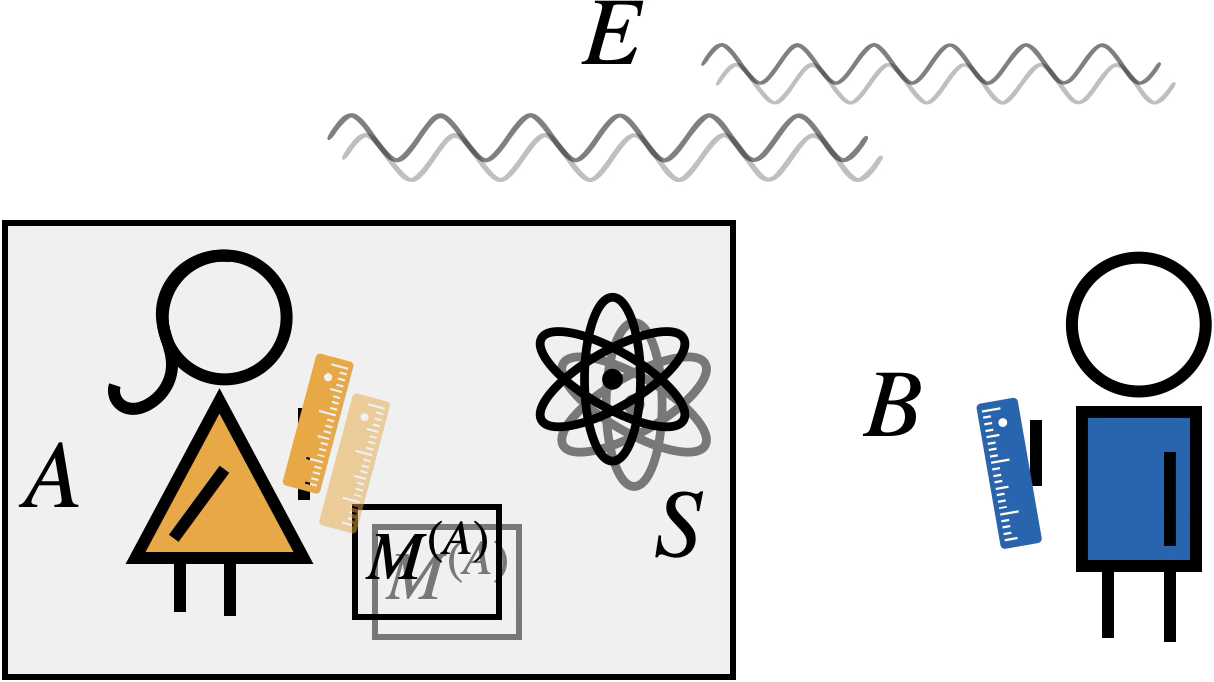}
          \caption{Bob's Perspective}
    \end{subfigure}
      \caption{\emph{Same physical interferometer.} Alice performs a Ramsey experiment $M^{(A)}$ on system $S$. Bob has access to $A$'s phase record and outcomes and describes the same physical interferometer from his point of view. The operational Ramsey visibility (and the associated rate) is frame-invariant: $\gamma_{nm}^{(B\,|\,\mathrm{access}\,A)}=\gamma_{nm}^{(A)}$. The shadowed images represent possible superpositions of relative configurations
between the systems and the reference frame.}
    \label{fig:Ramsey1}
\end{figure}

In the second scenario, illustrated in Fig.~\ref{fig:Ramsey2}, Bob performs a \emph{local} Ramsey experiment relative to his own phase standard and on the reduced state $\rho_S^{(B)}(t)$. This is not the same physical interferometer re-described in another frame, but a different operational procedure tied to a different reference system. In that case, the inferred decoherence rate is the one associated with the reduced state in frame $B$. In the Hamiltonian-symmetric regime, the difference between the locally inferred rates is exactly captured by the frame factor,
\begin{equation}
\gamma_{nm}^{(B)}(t)
=
\gamma_{nm}^{(A)}(t)+\gamma_{nm}^{\mathrm{fr}}(t),
\end{equation}
or, equivalently, at the level of visibilities,
\begin{equation}
\mathcal V_{nm}^{(B)}(t)
=
|F_{nm}(t)|\,\mathcal V_{nm}^{(A)}(t).
\label{eq:visibility_frame_factor}
\end{equation}
Thus, local Ramsey experiments in different frames need not return the same decoherence rate, not because the descriptions are inconsistent, but because they probe coherence relative to different physical phase standards.

\begin{figure}[H]
    \centering
    \begin{subfigure}{0.4\linewidth}
        \centering
          \includegraphics[width=\linewidth,trim=0 0 0 0,clip]{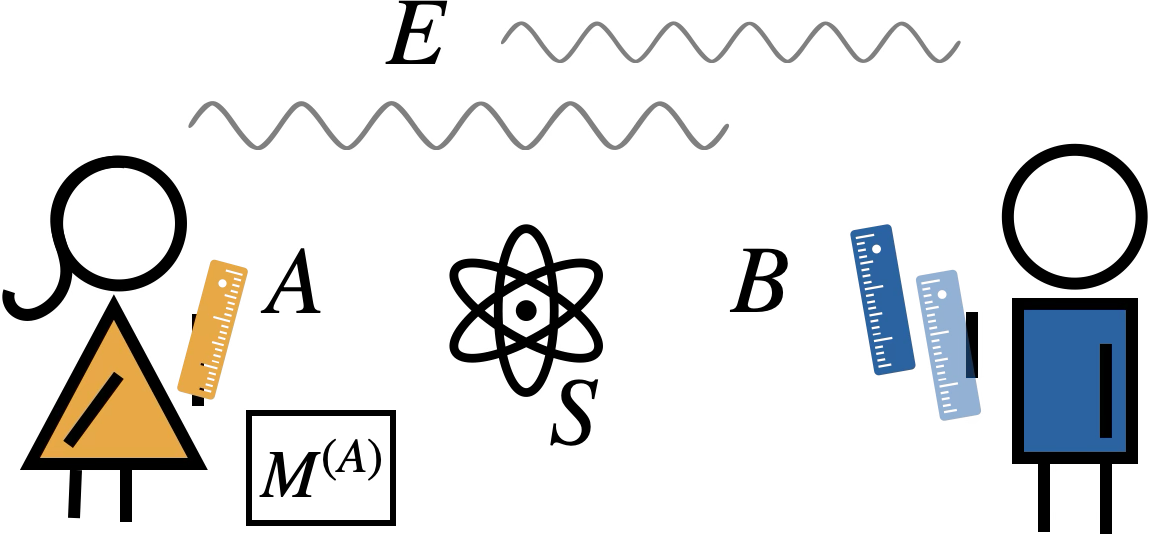}
    \caption{Alice's Perspective}
    \end{subfigure}\hspace{1.5cm}
   \begin{subfigure}{0.4\linewidth}
        \centering
          \includegraphics[width=0.9\linewidth,trim=0 0 0 0,clip]{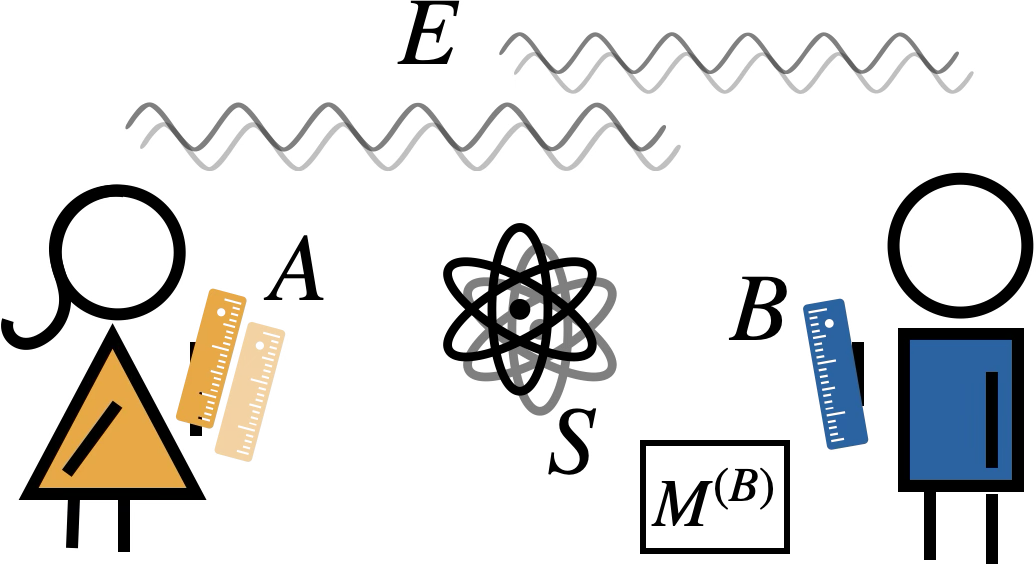}
    \caption{Bob's Perspective}
    \end{subfigure}
    \caption{\emph{Local Ramsey (different phase standards).} Observers Alice and Bob perform Ramsey
  experiments $M^{(A)}$ resp. $M^{(B)}$ using their own phase references $A$ and $B$. The inferred decoherence rates can differ. In the
   Hamiltonian-symmetric regime, the difference is captured by a frame factor, i.e.\ $\mathcal{V}_B(t)=|F_{nm}(t)|\,\mathcal{V}_A(t)$.}
    \label{fig:Ramsey2}
\end{figure}

This is the operational meaning of frame-induced decoherence. A fixed physical interferometer has frame-invariant statistics once both preparation and measurement are transformed consistently across frames. What changes, instead, is the reduced coherence seen by an observer who uses a different local reference to define and interrogate the same system.

\section{Case study: gravity-motivated dephasing with a degrading phase reference} \label{sec:Blencowe}
We now apply the reduced-channel framework to a gravity-motivated pure-dephasing
model in which the two contributions to the locally inferred decoherence rate
identified in Sec.~\ref{sec:Ramsey} can be separated explicitly. Several weak-gravity or gravity-inspired models lead, under suitable approximations and coarse graining, to an effective visibility loss for energy-sensitive superpositions that can be described as pure dephasing in an appropriate basis~\cite{Blencowe2013,Pikovski2015,Bonifacio2009,Balatsky2025}. At the same time, observing phase effects in these scenarios requires a phase/clock standard that is itself a physical system. In \cite{Fahn_2022}, for example, a geometrical clock, built from the gravitational degrees of freedom, was used to describe gravitationally induced decoherence in a relational manner. Treating the reference quantum mechanically raises the question of frame-induced contributions to the locally inferred dephasing, thus providing 
a natural testbed for the preserving criteria developed above. 

We will focus in particular on the scenario where the system
$S$ undergoes energy-basis dephasing in the laboratory frame due to its
coupling to a weak gravitational environment, and the change of QRF is implemented with respect to a phase reference whose state
becomes progressively delocalized over $\mathrm U(1)$. Thus, the locally inferred
decoherence rate in the transformed description acquires an additional
reference-induced contribution that is not generated by any new physical
interaction between $S$ and the reference, but reflects the operational use
of a degrading phase standard.

\paragraph{Physical setup.}
We consider four sets of degrees of freedom in the laboratory frame $A$: the
system $S$, a harmonic oscillator of frequency $\Omega$ with energy eigenstates
$\{\ket n\}$; the environment $E$, modeling a gravitational bath; the
phase reference $B$; and an auxiliary environment $E'$ that degrades the
reference. The total Hamiltonian is
\begin{equation}
H^{(A)}
= H_S^{(A)} + H_E^{(A)} + H_{SE}^{(A)} + H_B^{(A)} + H_{E'}^{(A)} + H_{BE'}^{(A)},
\label{eq:total_hamiltonian_blencowe}
\end{equation}
with no direct $S$--$B$ or $S$--$E'$ coupling. The system Hamiltonian is
$H_S^{(A)} = \hbar\Omega(N_S + 1/2)$ with $N_S = a^\dagger a$
so that the energy eigenvalues are indeed $\hbar\Omega\left(n+\frac{1}{2}\right)$. Following
gravity-inspired dephasing models, we take the $S$--$E$ interaction to be of
energy-coupling form,
\begin{equation}
H_{SE}^{(A)} = H_S^{(A)} \otimes B_E^{(A)},
\qquad
B_E^{(A)} = \sum_i \lambda_i q_i,
\label{eq:HSE_blencowe}
\end{equation}
with $q_i$ the position operators of the environment and $\lambda_i$ the corresponding coupling constants. Such a coupling
can be motivated from the matter--graviton interaction
$S_I \propto \int d^4x\,T_{\mu\nu}h^{\mu\nu}$ after suitable truncations and
coarse graining~\cite{Blencowe2013,XuBlencowe}. Since
$[H_S^{(A)}, H_{SE}^{(A)}] = 0$, the interaction does not exchange energy
between $S$ and $E$, and the reduced dynamics of $S$ in frame $A$ is of
pure-dephasing type.
 
For an initially factorized state $\rho_{SBEE'}^{(A)}(0)=
\rho_S^{(A)}(0)\otimes\rho_B^{(A)}(0)
\otimes\rho_E^{(A)}(0)\otimes\rho_{E'}^{(A)}(0)$ and a stationary Gaussian gravitational
bath,\footnote{Under these assumptions the second-order cumulant expansion is
exact and the dephasing functional takes the Gaussian form below.} the energy-basis coherences of \(S\) in the laboratory frame read
\begin{equation}
\rho_{nm}^{(A)}(t)
= e^{-i\omega_{nm}t}\,e^{-\Gamma_{nm}^{\rm grav}(t)}\,\rho_{nm}^{(A)}(0),
\qquad
\Gamma_{nm}^{\rm grav}(t)
= \Omega^2(n-m)^2\,\chi_{\rm grav}(t),
\label{eq:grav_dephasing_functional}
\end{equation}
where $\omega_{nm} := \Omega(n - m)$ and
$\chi_{\rm grav}(t)$ encapsulates all the model-dependent information about
the gravity-motivated environment (spectral density, temperature, cutoff,
Planck-scale suppression). The characteristic $(n - m)^2$ scaling is
forced by the energy-coupling form of the interaction. The corresponding
laboratory-frame rate is
\begin{equation}
\gamma_{nm}^{(A)}(t)
= -\frac{d}{dt}\ln|\rho_{nm}^{(A)}(t)|
= (n-m)^2\,\Omega^2\,\dot\chi_{\rm grav}(t).
\label{eq:grav_rate_A}
\end{equation}
A simple explicit benchmark is obtained in the long-time, high-temperature
Born--Markov regime of gravitational energy-dephasing models, where the
cutoff-dependent initial transient is neglected
\cite{Blencowe2013,XuBlencowe}. In this regime, the decoherence rate for a
superposition of two stationary energy states with gap \(\Delta E\) takes the form
\begin{equation}
\gamma_{\rm grav}^{\rm BM}(\Delta E)
=
\frac{k_B T_g}{\hbar}
\left(\frac{\Delta E}{E_P}\right)^2 ,
\label{eq:grav_BM_rate}
\end{equation}
where \(T_g\) is the effective temperature of the thermal graviton
environment and \(E_P=\sqrt{\hbar c^5/G}\) is the Planck energy. For the harmonic oscillator
considered here, substituting
\(\Delta E_{nm}=\hbar\Omega(n-m)\) into
Eq.~\eqref{eq:grav_BM_rate} gives the laboratory-frame Born--Markov rate
\begin{equation}
\gamma_{nm}^{(A),{\rm BM}}
=
\gamma_{\rm grav}^{\rm BM}(\Delta E_{nm})
=
(n-m)^2\,\Omega^2\,
\frac{\hbar k_B T_g}{E_P^2}.
\label{eq:grav_rate_A_BM}
\end{equation}

\paragraph{Phase QRF and degrading reference.}
We model $B$ as an ideal QRF for $G = \mathrm U(1)$, with Hilbert space
$\mathcal H_B^{(A)} = L^2(\mathrm U(1), d\theta/2\pi)$, orthogonal orientation
states $\{\ket{\theta}_B\}$, and group action on $S$ generated by the number
operator,
$V_S(\theta) = e^{-i\theta N_S}$, so that $V_S(\theta)\ket n = e^{-in\theta}\ket n$.
The QRF transformation is
\begin{equation}
\hat S_{A\to B}
= \int_0^{2\pi}\frac{d\theta}{2\pi}\,
\ket{2\pi-\theta}_A\!\bra{\theta}_B
\otimes V^{\dagger}_S(\theta)\otimes V^{\dagger}_E(\theta)\otimes \mathbb I_{E'} ,
\label{eq:QRF_U1_blencowe}
\end{equation}
where the group acts trivially on $E'$ for simplicity\footnote{Including a
non-trivial $V_{E'}(\theta)$ would not affect the reduced random-unitary action
on $S$, as long as the total group action factorizes.}. Since
$[V_S(\theta), H_S^{(A)}] = 0$, the transformation is Hamiltonian-symmetric and
Theorem~\ref{thm:gpreserving_preservation} applies, i.e., populations are preserved,
and coherences transform multiplicatively through a frame factor, see Sec. \ref{subsec:frame_factors}.
 
The orientation statistics of the reference are encoded in the diagonal of its
state in the orientation basis. Writing
$\rho_B^{(A)}(t) = \Tr_{E'}[\rho_{BE'}^{(A)}(t)]$, we define the (normalized)
phase distribution

\begin{equation}
p(\theta,t) := \frac{1}{2\pi}\bra{\theta}\rho_B^{(A)}(t)\ket{\theta},
\qquad
\int_0^{2\pi} d\theta\,p(\theta,t) = 1.
\label{eq:p_distribution_blencowe}
\end{equation}
Under the factorized-action assumption of
Sec.~\ref{subsec:Quantumvsclassical}, invariance of the partial trace over $E$
under local unitaries removes $V_E(\theta)$ from the reduced action, yielding
the random-unitary form
\begin{equation}
\rho_S^{(B)}(t)
= \int_0^{2\pi} d\theta\, p(\theta,t)\,V^\dagger_S(\theta)\,\rho_S^{(A)}(t)\,V_S(\theta).
\label{eq:rhoB_random_unitary_blencowe}
\end{equation}
 
We now model the degradation of the reference phenomenologically at the level
of the orientation distribution \(p(\theta,t)\). In particular, we assume that
degradation corresponds to a progressive spreading of \(p(\theta,t)\) on the
circle. In the small-spread regime, where the distribution is sharply localized
compared to \(2\pi\), this can be described by the Gaussian approximation \footnote{Formally, since \(\theta\) is an angular variable, this Gaussian should be
understood as the small-spread approximation to a wrapped distribution on
\(\mathrm U(1)\).}
\begin{equation}
p(\theta, t) =
\frac{1}{\sqrt{2\pi\sigma_\theta^2(t)}}
\exp\!\left[-\frac{(\theta - \theta_0(t))^2}{2\sigma_\theta^2(t)}\right]
\label{eq:p_gaussian_blencowe}
\end{equation}
The essential feature is that
diffusion produces a \emph{non-rigid} deformation of $p(\theta,t)$ on
$\mathrm U(1)$: a rigid unitary drift would only translate the distribution,
leaving the modulus of its Fourier components unchanged
(cf.~Appendix~\ref{sec:rigidtranslations}); broadening, instead, suppresses
these moduli and therefore the frame factors below.

Microscopically, such degradation may arise from the coupling of the reference
\(B\) to the auxiliary degrees of freedom \(E'\). Here we describe the resulting
reduced dynamics phenomenologically at the level of the phase distribution.
To make this effect explicit, we consider two simple laws for the spreading of
\(p(\theta,t)\). The first is diffusion on the group, modeled directly by the
heat equation on \(\mathrm{U}(1)\)~\cite{Risken1996,Gardiner2009},
\[
\partial_t p(\theta,t)
=
D_\theta\,\partial_\theta^2 p(\theta,t),
\]
with periodic boundary conditions. For an initially localized distribution,
the solution is a wrapped Gaussian; in the small-spread regime it is well
approximated by Eq.~\eqref{eq:p_gaussian_blencowe} with
\begin{equation}
\sigma_\theta^2(t)
=
\sigma_\theta^2(0)+2D_\theta t
\equiv
\sigma_\theta^2(0)+\Gamma_B t ,
\label{eq:linear_phase_diffusion}
\end{equation}
where \(\Gamma_B=2D_\theta\) quantifies the phase-diffusion rate of the
reference. The second is a ballistic broadening model,
\begin{equation}
\sigma_\theta^2(t)
=
\sigma_\theta^2(0)+\kappa_B t^2 .
\label{eq:ballistic_phase_spreading}
\end{equation}
This law can be understood as arising from quasi-static frequency fluctuations
of the reference. If, during each experimental run, the reference phase evolves
as $\theta(t)=\theta(0)+\delta\omega\, t$, with \(\delta\omega\) approximately constant during the run but sampled from
run to run from a distribution of variance
\(\mathrm{Var}(\delta\omega)=\kappa_B\), then averaging over the frequency
offsets gives
\[
\sigma_\theta^2(t)
=
\sigma_\theta^2(0)+\mathrm{Var}(\delta\omega)\,t^2 .
\]
The ballistic model can be understood as the analogue, for the reference phase,
of inhomogeneous dephasing due to slowly varying or shot-to-shot frequency
fluctuations~\cite{Ithier2005,Cywinski2008,Bylander2011}.

In both cases, the drift \(\theta_0(t)\) affects only the phase of the
relevant Fourier components, whereas \(\sigma_\theta^2(t)\) controls their
moduli and hence the reference-induced contribution to the locally inferred
decoherence rate.

\paragraph{Frame-induced contribution to the rate.}
Combining Eq.~\eqref{eq:rhoB_random_unitary_blencowe} with
$V_S(\theta)\ket n = e^{-in\theta}\ket n$ leaves populations unchanged and gives
\begin{equation}
\rho_{nm}^{(B)}(t) = F_{nm}(t)\,\rho_{nm}^{(A)}(t),
\qquad
F_{nm}(t) = \int_0^{2\pi} d\theta\, p(\theta, t)\,e^{+i(n-m)\theta},
\label{eq:Fnm_blencowe}
\end{equation}
where $F_{nm}(t)$ is the $(n-m)$-th Fourier component of $p(\theta,t)$. For
the Gaussian distribution~\eqref{eq:p_gaussian_blencowe}, one has
\begin{equation}
|F_{nm}(t)| = \exp\!\left[-\frac{(n-m)^2}{2}\,\sigma_\theta^2(t)\right].
\label{eq:Fnm_modulus_blencowe}
\end{equation}
Substituting into the additive decomposition of
Sec.~\ref{subsec:frame_factors} and using Eq.~\eqref{eq:grav_rate_A}, the
locally inferred decoherence rate in frame $B$ becomes
\begin{equation}
\boxed{
\gamma_{nm}^{(B)}(t)
= \underbrace{(n-m)^2\,\Omega^2\,\dot\chi_{\rm grav}(t)}_{\text{gravity-induced}}
+ \underbrace{\frac{(n-m)^2}{2}\,\frac{d\sigma_\theta^2(t)}{dt}}_{\text{reference-induced}}.
}
\label{eq:grav_reference_rate}
\end{equation}
The laboratory-frame rate
is fixed by the gravity-motivated environment \(E\), while the additional term
is determined entirely by the dynamics of the phase reference \(B\). For a
linearly increasing phase variance,
the reference contribution reduces to a constant additive shift
$\gamma_{nm}^{\rm fr}
=
\frac{(n-m)^2}{2}\Gamma_B $, while a ballistic phase spreading
gives $\gamma_{nm}^{\rm fr}(t)
=
(n-m)^2\kappa_B t $.
Thus, even if the environment-induced rate is constant in frame \(A\), the rate
locally inferred in frame \(B\) can acquire either a constant offset or a
genuinely time-dependent contribution from the degrading phase reference.

\begin{figure}[H]
    \centering
    \begin{subfigure}[b]{0.48\textwidth}
        \centering
        \includegraphics[width=\textwidth]{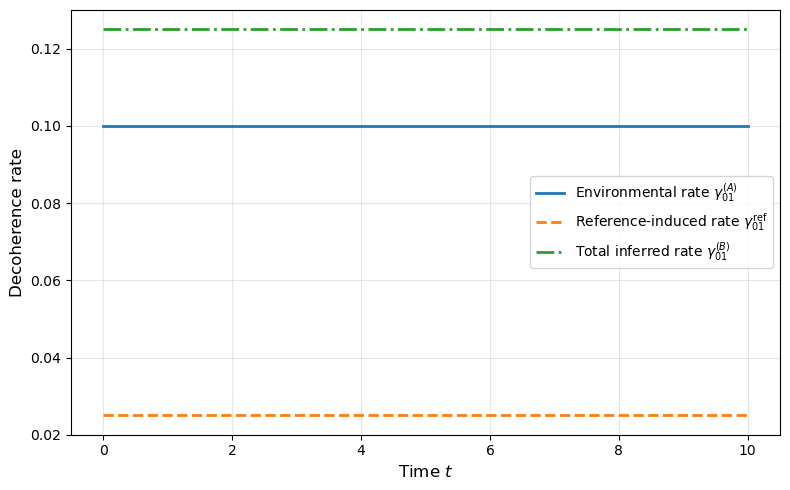}
        \caption{Linear phase diffusion.}
        \label{fig:rate-decomposition-blencowe}
    \end{subfigure}
    \hfill
    \begin{subfigure}[b]{0.48\textwidth}
        \centering
        \includegraphics[width=\textwidth]{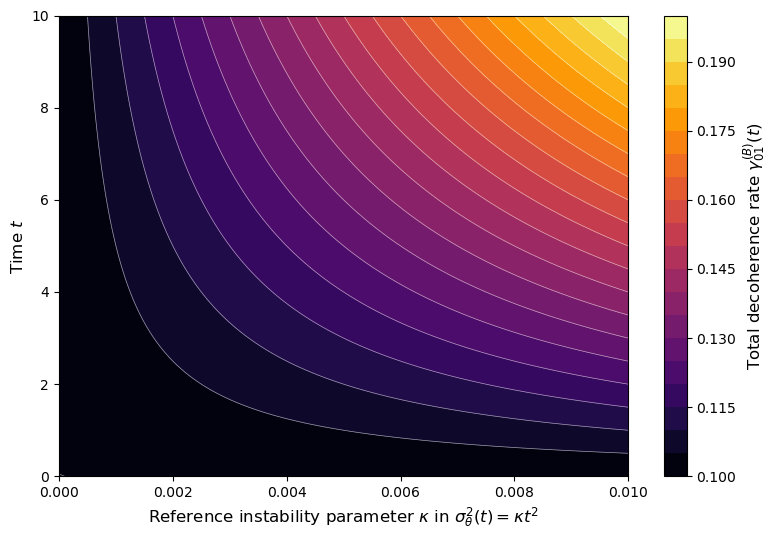}
        \caption{Ballistic phase spreading.}
        \label{fig:contourplot-blencowe}
    \end{subfigure}
    \caption{
    Reference-induced contributions to the locally inferred decoherence rate in
    the gravity-motivated pure-dephasing case study. Left: for linear
    phase-reference degradation, \(\sigma_\theta^2(t)=\Gamma_B t\), the
    reference adds a constant contribution
    \(\gamma_{01}^{\rm fr}=\Gamma_B/2\), so that
    \(\gamma_{01}^{(B)}=\gamma_{01}^{(A)}+\gamma_{01}^{\rm fr}\).
    Right: for ballistic phase spreading,
    \(\sigma_\theta^2(t)=\kappa_B t^2\), the reference-induced contribution is
    time dependent, \(\gamma_{01}^{\rm fr}(t)=\kappa_B t\), and the inferred
    rate grows with both time and the instability parameter \(\kappa_B\).
    }
    \label{fig:decoherence-comparison}
\end{figure}

The decomposition~\eqref{eq:grav_reference_rate} realizes the general
distinction of Sec.~\ref{sec:Ramsey}: the gravitational environment fixes
$\gamma_{nm}^{(A)}(t)$, while degradation of the phase reference contributes
an additional, independently controllable term to the locally inferred rate.
The setting is operationally delicate: in the harmonic-oscillator model, both
contributions scale as $(n-m)^2$, so that a quadratic energy-gap dependence of
Ramsey visibility decay is not, by itself, a unique signature of
gravitational decoherence: a degrading phase reference can mimic, mask, or
renormalize the same scaling. In this sense, the reference-induced contribution plays a role similar to that of competing environmental decoherence mechanisms that must be disentangled from a putative gravitational contribution in realistic interferometric settings~\cite{CarlessoBassi2016}. Here, however, the competing effect does not arise from an additional physical interaction with the system, but from the degradation of the phase reference relative to which coherence is defined and measured.

Disentangling the two effects requires independent
control or calibration of the reference, for example by varying the
phase-diffusion rate $\Gamma_B$, comparing different reference-stabilization
protocols, or exploiting parameters that affect $\chi_{\rm grav}(t)$ but not
$\sigma_\theta^2(t)$, such as the system frequency, the gravitational bath
temperature, or the environmental cutoff. Beyond the Gaussian regime, the
Fourier structure of $F_{nm}(t)$ provides further tools to identify the actual origin of dephasing: departures from
a quadratic-in-$(n-m)$ profile, non-Gaussian time dependence, or coherence
revivals would indicate non-Gaussian reference noise, breakdown of the
diffusion model, or non-classical reference effects beyond the
classical-misalignment regime of Sec.~\ref{subsec:Quantumvsclassical}. The
reduced QRF channel framework therefore shows that Ramsey visibility is not
determined solely by the gravitational environment, but also by the phase
reference relative to which coherence is defined and measured.

\section{Discussion and outlook}
\label{sec:conclusions}
We have developed a reduced-channel description of QRF transformations for open
quantum systems, focusing on the information that remains accessible after the
reference frame and auxiliary degrees of freedom have been discarded. This
construction makes it possible to investigate which features of a reduced open-system
description are robust under a change of QRF, and which ones instead depend on
the physical reference used to define and probe the subsystem. This distinction
is physically relevant because qualitative properties of a reduced dynamics,
such as population preservation, pure-dephasing behavior, or
decoherence rates, are often interpreted as signatures of the underlying
system--environment interaction. Our results show how to separate such
interaction-induced features from effects introduced by the quantum reference
itself.

First, we introduced a hierarchy of preservation conditions for reduced QRF
channels. Block-preservation ensures that populations in the transformed
frame are insensitive to inter-block coherences in the original frame, although
they may still depend on inaccessible reference and environmental degrees of
freedom. Strong block-preservation identifies the stricter case in which each
output block is associated with a unique input block, independently of the state
of the inaccessible degrees of freedom. These notions are purely kinematic:
they are properties of the reduced QRF channel itself and do not yet depend on
the particular open-system dynamics.
When the dynamics in the original frame is of pure-dephasing type, we
showed that kinematic block-preservation alone is not sufficient to guarantee
that the pure-dephasing population structure survives the change of QRF. The
missing ingredient is a dynamical compatibility condition between the reduced
QRF channel and the controlled evolution generating dephasing in the original
frame, which provides a necessary and sufficient condition for population preservation
for the full class of factorized initial preparations considered here. 

We further identified the regimes in which a change of quantum reference frame induces a reduced system transformation that admits a classical interpretation as random frame misalignment. In these regimes, whenever the original reduced dynamics is unital, non-unitality of the transformed reduced dynamics provides a sufficient witness of reduced QRF effects that cannot be explained by classical uncertainty over frame orientations alone.

Finally, in the Hamiltonian-symmetric regime, we showed that the frame
dependence of decoherence can be given a direct operational interpretation. In
this case, the transformed coherences acquire a multiplicative frame factor,
and the locally inferred decoherence rate splits into an environment-induced and a reference-induced contribution. Ramsey interferometry makes this
decomposition observable: the same physical interferometer yields
frame-independent statistics when preparation and measurement are transformed
consistently, whereas local Ramsey experiments tied to different phase
standards can infer different decoherence rates. This illustrates how a
reference-induced loss of visibility can be separated from ordinary
environment-induced decoherence.
We further investigated this point in a gravity-motivated model of decoherence, for which we showed that
both the environmental and the reference-induced
contributions scale quadratically with the energy gap. A measured dephasing rate with this energy-gap dependence is
therefore not, by itself, a unique signature of gravitational decoherence: a
degrading phase reference can mimic the same scaling. Distinguishing the two
contributions requires an independent characterization or calibration of the
reference.

These results point to several directions for future work. A first direction is to develop a resource-oriented view of non-ideal quantum reference frames. Rather than treating non-ideality only as a source of error, one may ask which finite reference resources, limited resolution, correlations with discarded degrees of freedom, or intrinsic reference degradation enable reduced transformations that depart from the classical-misalignment class. In this sense, non-ideal QRFs may enlarge the set of reachable reduced open-system dynamics and provide operational witnesses of genuinely quantum reference-frame behavior. This complements existing work on non-ideal QRFs \cite{delaHamette_2021_perspectiveneutral,Garmier_2025,delaHamette_entropy_2026}, where frame imperfections are known to modify conservation statements under QRF changes.

A second direction is to extend the preservation criteria beyond pure
dephasing. Phase-covariant dynamics provide a natural first step, since they
retain a distinguished coherence--population structure \cite{Chruscinski2022}. Amplitude damping,
thermalization, and general Lindblad dynamics would require identifying the
appropriate invariant objects, such as fixed-point algebras or
decoherence-free subspaces, to be tested under the reduced QRF channel.

A third direction is to study QRF transformations with non-factorizing group
actions, as in relativistic spin--momentum settings \cite{Giacomini2018} or field-theoretic
contexts. In such cases, a change of frame can entangle accessible and
inaccessible degrees of freedom in ways that again fall outside the
classical-misalignment class, leading to a richer family of reduced QRF
channels and to new forms of reference-dependent open-system dynamics.\\

More broadly, changing quantum reference frame and then reducing
to the accessible system can change which noise mechanisms are locally inferred, which quantities are
preserved, and how observed loss of coherence or relaxation should be
attributed between environmental interactions and limitations of the reference
used to define phase, basis, or energy. Reduced QRF channels provide
a systematic way to separate intrinsic features of an interaction from
reference-dependent features of its operational description. This distinction
should be relevant in settings where observers, subsystems, and
observables cannot be separated a priori, from relational approaches to quantum
gravity \cite{Tambornino2011} to observer-dependent descriptions in quantum field theory \cite{Witten_2021,Chandrasekaran_2022,DeVuyst_2024,DeVuyst_2025}.

\begin{acknowledgments}
We would like to thank the QIT group at ETH for useful discussions and feedback.
This work was supported by STSM Grants from COST Action CA23115: Relativistic Quantum Information (RQI), and COST Action CA23130: Bridging high and low energies in search of quantum gravity (BridgeQG). Both actions are funded by COST (European Cooperation in Science and Technology).
V.K. thanks the Swiss National Science Foundation via the NCCR SwissMAP and the ETH Zurich Quantum Center for support.
F.G. acknowledges support from the Swiss National Science Foundation via the Ambizione Grant PZ00P2-208885, from the ETH Zurich Quantum Center, and from the Italian Ministry of University and Research (MUR) under the Funding for Research Projects (FIS 2), pursuant to Ministerial Decree No. 23314 of 11/12/2024 (Project Q-GraSp, FIS-2023-02629). 
This work was made possible through the support of the WOST, WithOut SpaceTime project (https://withoutspacetime.org), led by the Center for Spacetime and the Quantum (CSTQ), and supported by Grant ID\#~63683 from the John Templeton Foundation (JTF). The opinions expressed in this work are those of the author(s) and do not necessarily reflect the views of the John Templeton Foundation or those of the Center for Spacetime and the Quantum. 
\end{acknowledgments}

\section*{Author Contributions}
P.L. and A.S. conceived the project, developing the concrete research questions in discussion with V.K. and F.G.. P.L. carried out the majority of the analytical and computational work and prepared the first draft of the manuscript under the local supervision of V.K.. All authors contributed to the development and interpretation of the results, provided feedback throughout the project, and reviewed and revised the manuscript.

\appendix
\numberwithin{equation}{section}
\renewcommand{\theequation}{\thesection\arabic{equation}}
\section{Proof of Lemma~\ref{lem:pinching}}
\label{App:equivalencepinching}

In this appendix, we prove Lemma \ref{lem:pinching}, which we repeat here for convenience:

\begin{lemma*}[Pinching form of block preservation, rank-one case]
Assume the block structure is non-degenerate (rank-$1$), then $\mathcal Q$ is block-preserving if and only if 
\begin{equation}
\widetilde\Delta^{(B)}_{S}\circ\mathcal Q
=
\widetilde\Delta^{(B)}_{S}\circ\mathcal Q\circ\Delta_{SBE}^{(A)}.
\label{eq:diagonal_pinching_intertwining_app}
\end{equation}
Moreover, if
\begin{equation}
\widetilde\Delta^{(B)}_{S}\circ\mathcal Q
=
\mathcal Q\circ\Delta_{SBE}^{(A)}
\label{eq:full_pinching_intertwining_app}
\end{equation}
then $\mathcal Q$ is block-preserving in the rank-$1$ case.
\end{lemma*}

Let $\{P_m^{(A)}\}_m$ be a PVM on $\mathcal H_S^{(A)}$ and
$\{\widetilde P_n^{(B)}\}_n$ a PVM on $\mathcal H_S^{(B)}$. Define
\[
\widetilde\Delta_S^{(B)}(Y):=\sum_n \widetilde P_n^{(B)}Y\widetilde P_n^{(B)},
\qquad
\Delta_{SBE}^{(A)}(X):=\sum_m (P_m^{(A)}\!\otimes\!\mathbb I_{BE})\,X\,(P_m^{(A)}\!\otimes\!\mathbb I_{BE}).
\]
Let $\mathcal Q:\mathcal T(\mathcal H_{SBE}^{(A)})\to \mathcal T(\mathcal H_S^{(B)})$ be CPTP, and let $\mathcal Q^\dagger$ be its dual map, defined by
\[
\Tr_S\!\big[Y\,\mathcal Q(X)\big]
=
\Tr_{SBE}\!\big[\mathcal Q^\dagger(Y)\,X\big]
\qquad
\forall X\in\mathcal T(\mathcal H_{SBE}^{(A)}),\ \forall Y\in\mathcal B(\mathcal H_S^{(B)}).
\]

\medskip
\noindent\textbf{Claim.}
Assume that the chosen block structures are non-degenerate (rank-$1$), i.e.\
each $P_m^{(A)}$ and each $\widetilde P_n^{(B)}$ has rank $1$. Then the following are equivalent:
\begin{align}
\widetilde\Delta_S^{(B)}\circ\mathcal Q
&=
\widetilde\Delta_S^{(B)}\circ\mathcal Q\circ\Delta_{SBE}^{(A)},
\label{eq:weak_pinching_intertwining_app_compact}\\
\mathcal Q^\dagger(\widetilde P_n^{(B)})
&=
\sum_m P_m^{(A)}\otimes \Pi_{nm}^{(A)},
\qquad
\Pi_{nm}^{(A)}\ge 0,\qquad
\sum_n \Pi_{nm}^{(A)}=\mathbb I_{BE}^{(A)}\ \ \forall m.
\label{eq:block_alignment_heis_app_compact}
\end{align}

\noindent\emph{Proof.}
Assume \eqref{eq:weak_pinching_intertwining_app_compact}. Using trace duality, together with
\[
\Tr_S[Y\,\widetilde\Delta_S^{(B)}(Z)]
=
\Tr_S[\widetilde\Delta_S^{(B)}(Y)\,Z],
\qquad
\Tr_{SBE}[W\,\Delta_{SBE}^{(A)}(X)]
=
\Tr_{SBE}[\Delta_{SBE}^{(A)}(W)\,X],
\]
we obtain
\[
\mathcal Q^\dagger\circ\widetilde\Delta_S^{(B)}
=
\Delta_{SBE}^{(A)}\circ\mathcal Q^\dagger\circ\widetilde\Delta_S^{(B)}.
\]
Applying this identity to $\widetilde P_n^{(B)}$ gives
\[
\mathcal Q^\dagger(\widetilde P_n^{(B)})
=
\Delta_{SBE}^{(A)}\!\big(\mathcal Q^\dagger(\widetilde P_n^{(B)})\big),
\]
so $\mathcal Q^\dagger(\widetilde P_n^{(B)})$ is block-diagonal with respect to
$\{P_m^{(A)}\otimes\mathbb I_{BE}\}_m$. Since $P_m^{(A)}$ has rank $1$, this yields
\[
\mathcal Q^\dagger(\widetilde P_n^{(B)})=\sum_m P_m^{(A)}\otimes \Pi_{nm}^{(A)}.
\]
Positivity of the $\Pi_{nm}^{(A)}$ follows from complete positivity of $\mathcal Q^\dagger$.
Moreover, since $\mathcal Q$ is trace-preserving, $\mathcal Q^\dagger$ is unital, hence
\[
\mathbb I_{SBE}^{(A)}
=
\mathcal Q^\dagger(\mathbb I_S^{(B)})
=
\sum_n \mathcal Q^\dagger(\widetilde P_n^{(B)})
=
\sum_m P_m^{(A)}\otimes\Big(\sum_n \Pi_{nm}^{(A)}\Big),
\]
which implies $\sum_n\Pi_{nm}^{(A)}=\mathbb I_{BE}^{(A)}$ for every $m$.

Conversely, assume \eqref{eq:block_alignment_heis_app_compact}. Then each
$\mathcal Q^\dagger(\widetilde P_n^{(B)})$ is invariant under $\Delta_{SBE}^{(A)}$.
Since the $\widetilde P_n^{(B)}$ are rank-$1$, every operator in the range of
$\widetilde\Delta_S^{(B)}$ has the form $\widetilde\Delta_S^{(B)}(Y)=\sum_n c_n(Y)\,\widetilde P_n^{(B)}.$
By linearity,
\[
\Delta_{SBE}^{(A)}\circ\mathcal Q^\dagger\circ\widetilde\Delta_S^{(B)}
=
\mathcal Q^\dagger\circ\widetilde\Delta_S^{(B)}.
\]
Passing again to the dual maps gives
\[
\widetilde\Delta_S^{(B)}\circ\mathcal Q
=
\widetilde\Delta_S^{(B)}\circ\mathcal Q\circ\Delta_{SBE}^{(A)},
\]
which is \eqref{eq:weak_pinching_intertwining_app_compact}. \hfill$\square$

\medskip
\noindent\textbf{Corollary.}
If
\begin{equation}
\widetilde\Delta_S^{(B)}\circ\mathcal Q
=
\mathcal Q\circ\Delta_{SBE}^{(A)},
\label{eq:strong_pinching_intertwining_app_compact}
\end{equation}
then $\mathcal Q$ is block-preserving in the sense of
\eqref{eq:block_alignment_heis_app_compact}.

\emph{Proof.}
Compose \eqref{eq:strong_pinching_intertwining_app_compact} on the left with
$\widetilde\Delta_S^{(B)}$ and use idempotence of $\widetilde\Delta_S^{(B)}$ to recover
\eqref{eq:weak_pinching_intertwining_app_compact}; the claim then applies. \hfill$\square$

\medskip
\noindent\emph{Remarks.}
\begin{enumerate}[label=(\roman*)]
\item For degenerate blocks, \eqref{eq:weak_pinching_intertwining_app_compact} still implies block-diagonality with respect to $\{P_m^{(A)}\otimes\mathbb I_{BE}\}_m$, but the factorized form $P_m^{(A)}\otimes\Pi_{nm}^{(A)}$ generally requires an additional blockwise alignment condition.

\item If $\widetilde P_n^{(B)}(t)$, and hence $\widetilde\Delta_S^{(B)}(t)$, depend on time, the equivalence holds pointwise in time:
\[
\widetilde\Delta_S^{(B)}(t)\circ\mathcal Q_t
=
\widetilde\Delta_S^{(B)}(t)\circ\mathcal Q_t\circ\Delta_{SBE}^{(A)}
\quad\Longleftrightarrow\quad
\eqref{eq:block_alignment_heis_app_compact}
\ \text{with }\Pi_{nm}^{(A)}(t).
\]
Additional consistency conditions arise only if one requires a single fixed block structure for all $t$.
\end{enumerate}

\section{Rigid translations on groups and modulus preservation}\label{sec:rigidtranslations}
In this appendix, we show that a rigid drift of the reference distribution on
the group changes only the phase of the frame factor and therefore gives no
reference-induced contribution to the decoherence rate. Throughout the
appendix, \(p_t(g)\) denotes the probability density of the reference
orientation with respect to the left Haar measure \(d\mu(g)\), normalized as $\int_G p_t(g)\, d\mu(g) = 1$.

\begin{definition}[Rigid translation of the reference distribution]\label{def:rigid}
A time-dependent family of probability densities $p_t(g)$ is a \emph{rigid (left) translation} of a fixed density $p_0(g)$
if there exists a (measurable) path $h:[0,\infty)\to G$ such that
\begin{equation}\label{eq:rigid-kernel}
p_t(g)\;=\;p_0\!\big(h(t)^{-1}g\big)\qquad\text{for all }g\in G.
\end{equation}
Equivalently, $p_t=L_{h(t)}p_0$, where $(L_hf)(g):=f(h^{-1}g)$ denotes left translation.
\end{definition}
Due to the left invariance of the Haar measure \cite{Robert1983},
\begin{equation}
    \int_G f(h^{-1}g)\,d\mu(g)
= \int_G f(g')\,d\mu(hg')
= \int_G f(g')\,d\mu(g')
= \int_G f(g)\,d\mu(g)
\end{equation}
for all $h\in G$. Hence \eqref{eq:rigid-kernel} describes a \emph{rigid motion} (no distortion) of the sampling profile on $G$.

\smallskip
We now want to show that under this condition, a rigid translation of the kernel can only change the phase of the frame factor  
\begin{equation}
    F_{nm}(t) = \int d\mu(g)\,p_t(g) e^{+i(\phi_n(g) -\phi_m(g))} := \int d\mu(g) \,p_t(g)\varphi_{nm}(g)
\end{equation}
while preserving its modulus. To do so, we will make use of an important property of $\varphi_{nm}(g)$.

\begin{definition}[Character condition]\label{def:no-shear}
Let $G$ be a (locally compact) group. A function $\varphi:G\to\mathbb C$ is said to satisfy the
\emph{character condition} if there exists a character $\chi:G\to\mathbb C^\times$ such that
\begin{equation}\label{eq:character-eig}
\varphi(hg)\;=\;\chi(h)\,\varphi(g)\qquad\text{for all }h,g\in G.
\end{equation}
\end{definition}
If $G$ is compact and $\chi:G\to\mathbb C^\times$ is a continuous character (with $\mathbb C^\times$ the multiplicative group of nonzero complex numbers), then $|\chi(h)|=1$
for all $h\in G$, i.e.\ $\chi$ is automatically unitary.

\smallskip 
We are working under the assumption that $V_S(g)|n\rangle = e^{-i\phi_n(g)}|n\rangle$. The fact that $V_S(g)$ is a representation of the group implies that $\phi_n(hg) = \phi_n(h) + \phi_n(g)$ and thus
\begin{equation}
    \varphi_{nm}(hg) = e^{+i(\phi_n(h)-\phi_m(h))}\varphi_{nm}(g).
\end{equation}
Hence, $\varphi_{nm}(g)$ satisfies the character condition with unitary character $\chi$. The key observation below shows that a rigid translation can at most add a global phase to $F_{nm}(t)$, given that the coherence function $\varphi_{nm}$ transforms as an eigenfunction of the left-regular action.

\begin{lemma}[Rigid translations preserve the modulus for unitary characters]\label{lem:rigid-preserve-modulus}
Let $G$ be a (locally compact) group with left Haar measure $\mu$.
Let $p_t(g)$ be a rigid translation of $p_0(g):=p(g,0)$ as in \eqref{eq:rigid-kernel}, and define
\[
F(t)\;:=\;\int_G p_t(g)\,\varphi(g)\,d\mu(g).
\]
Assume that $\varphi:G\to\mathbb C$ satisfies the character condition
\eqref{eq:character-eig}, i.e.\ there exists a character
$\chi:G\to\mathbb C^\times$ such that
\[
\varphi(hg)\;=\;\chi(h)\,\varphi(g)\qquad\text{for all }h,g\in G.
\]
Then
\begin{equation}\label{eq:F-chi}
F(t)\;=\;\chi\!\big(h(t)\big)\,\underbrace{\int_G p_0(g)\,\varphi(g)\,d\mu(g)}_{F(0)}.
\end{equation}
If, in addition, $\chi$ is unitary (i.e.\ $|\chi(h)|=1$ for all $h\in G$)---for example,
this holds automatically for continuous characters on compact groups---then
$|F(t)| = |F(0)|$ for all $t$.
\end{lemma}

\begin{proof}
Using \eqref{eq:rigid-kernel} and the left invariance of $d\mu(g)$ we obtain
\begin{align*}
F(t)
&= \int_G p_t(g)\,\varphi(g)\,d\mu(g)
 = \int_G p_0\big(h(t)^{-1}g\big)\,\varphi(g)\,d\mu(g).
\end{align*}
Perform the change of variables $g' = h(t)^{-1}g$, so that $g = h(t)g'$.
By left invariance, $d\mu(g) = d\mu(h(t)g') = d\mu(g')$, hence
\[
F(t)
= \int_G p_0(g')\,\varphi\big(h(t)g'\big)\,d\mu(g').
\]
Using the character condition \eqref{eq:character-eig} with $h=h(t)$ we get
\[
\varphi\big(h(t)g'\big) = \chi\big(h(t)\big)\,\varphi(g'),
\]
so that
\[
F(t)
= \chi\big(h(t)\big)\int_G p_0(g')\,\varphi(g')\,d\mu(g')
= \chi\big(h(t)\big)\,F(0),
\]
which proves \eqref{eq:F-chi}. If $\chi$ is unitary, then $|\chi(h(t))| = 1$ for all $t$, and
therefore $|F(t)| = |F(0)|$.
\end{proof}
In particular, for compact abelian groups every continuous character is unitary \cite{Rudin1962},
so rigid translations only contribute a global phase and do not affect the modulus of $F(t)$, and hence leave the decoherence rates
$\gamma_{nm}(t)$ invariant.

\section{Extension: Degenerate Blocks and Time-Dependent Reference Distributions}
\label{app:degenerate_blocks}

In this appendix we record a minimal extension of the frame-factor analysis to
degenerate energy blocks. We remain in the Hamiltonian-symmetric,
random-unitary regime considered in the main text, where the reduced action
induced on \(S\) by the QRF change takes the random-unitary form
\begin{equation}
\rho_S^{(B)}(t)
=
\int_G dg\, p_t(g)\,
V^{\dagger}_S(g)\rho_S^{(A)}(t)V_S(g),
\label{eq:app_random_unitary}
\end{equation}
with \(p_t(g)\ge 0\) and \(\int_G dg\,p_t(g)=1\). The time dependence of
\(p_t(g)\) describes the evolution, drift, or degradation of the reference
state, while the QRF transformation itself is kept fixed.

Assume that the group representation preserves the energy blocks,
\begin{equation}
[V_S(g),P_n]=0
\qquad
\forall g\in G,\ \forall n .
\label{eq:app_block_preserving_representation}
\end{equation}
If the spectral subspace \(P_n\mathcal H_S\) has dimension \(d_n>1\), choose an
orthonormal basis \(\{\ket{n,\alpha}\}_{\alpha=1}^{d_n}\) of
\(P_n\mathcal H_S\). Then the representation decomposes as
\begin{equation}
V_S(g)
=
\bigoplus_n V^{(n)}(g),
\label{eq:app_block_diagonal_representation}
\end{equation}
where \(V^{(n)}(g)\) acts on the internal degrees of freedom of the \(n\)-th
energy block.

For \(n\neq m\), the off-diagonal block of the reduced state transforms as
\begin{equation}
\langle n,\alpha|\rho_S^{(B)}(t)|m,\beta\rangle
=
\sum_{\alpha',\beta'}
M^{(n,m)}_{\alpha\alpha',\beta\beta'}(t)\,
\langle n,\alpha'|\rho_S^{(A)}(t)|m,\beta'\rangle,
\label{eq:app_degenerate_coherence_block}
\end{equation}
where
\begin{equation}
M^{(n,m)}_{\alpha\alpha',\beta\beta'}(t)
:=
\int_G dg\, p_t(g)\,
\bigl(V^{(n)}_{\alpha'\alpha}(g)\bigr)^*
V^{(m)}_{\beta'\beta}(g) .
\label{eq:app_degenerate_frame_map}
\end{equation}
Thus, for degenerate energy levels, the scalar frame factor \(F_{nm}(t)\) of
the non-degenerate case is replaced by a linear map \(M^{(n,m)}(t)\) acting on
the coherence block \(P_n\rho_S P_m\). When \(d_n=d_m=1\), this reduces to
\begin{equation}
F_{nm}(t)
=
\int_G dg\, p_t(g)\,
e^{+i[\phi_n(g)-\phi_m(g)]}.
\label{eq:app_scalar_frame_factor}
\end{equation}

When \([V_S(g),P_n]\neq 0\), the frame-factor description breaks down and the
full reduced QRF channel framework of Sec.~\ref{sec:qrf_preserving} is needed.\\

Finally, departing from the fixed-transformation assumption above, one may also
consider a prescribed time dependence in the system part of the QRF
transformation while remaining in the non-degenerate, Hamiltonian-symmetric
regime,
\begin{equation}
V_S(g,t)\ket n
=
e^{-i\phi_n(g,t)}\ket n .
\label{eq:app_time_dependent_representation}
\end{equation}
In this case the scalar frame factor becomes
\begin{equation}
F_{nm}(t)
=
\int_G dg\, p_t(g)\,
e^{+i\Delta\phi_{nm}(g,t)},
\qquad
\Delta\phi_{nm}(g,t)
:=
\phi_n(g,t)-\phi_m(g,t).
\label{eq:app_time_dependent_frame_factor}
\end{equation}
Differentiating gives
\begin{equation}
\dot F_{nm}(t)
=
\int_G dg\,
\left[
\partial_t p_t(g)
+
i\,p_t(g)\,\partial_t\Delta\phi_{nm}(g,t)
\right]
e^{+i\Delta\phi_{nm}(g,t)} .
\label{eq:app_time_dependent_frame_factor_derivative}
\end{equation}
The first term comes from the evolution of the reference-state distribution.
The second term appears only when the transformation rule itself has a
prescribed time dependence. If
\(\partial_t\Delta\phi_{nm}(g,t)\) is independent of \(g\), it changes only the
overall phase of \(F_{nm}(t)\), consistently with the modulus-preservation
property of rigid translations established in
Lemma~\ref{lem:rigid-preserve-modulus}. If it depends on \(g\), it can change
\(|F_{nm}(t)|\) and therefore contribute to the reference-induced decoherence
rate.

\section{Technical details of the Ramsey protocol}\label{ap:ramasey protocol}
This appendix provides technical details on the Ramsey protocol used in 
Sec.~\ref{sec:Ramsey} to operationally access decoherence rates, and clarifies 
how the protocol transforms under QRF changes.

\subsection{Operational estimation of decoherence rates: Ramsey protocol}

The decoherence rate $\gamma_{nm}(t)$, defined in Eq. (\ref{eq:decoherence_rate}), can be estimated
operationally via a Ramsey-type interferometric protocol on any two-level subspace
spanned by energy eigenstates $\{|n\rangle,|m\rangle\}$.

We regard the $\{|n\rangle,|m\rangle\}$ subspace as an effective qubit with Pauli
operators
\begin{align}
  \sigma_x^{(nm)} &= |n\rangle\langle m| + |m\rangle\langle n|, \\
  \sigma_y^{(nm)} &= -i|n\rangle\langle m| + i|m\rangle\langle n|, \\
  \sigma_z^{(nm)} &= |n\rangle\langle n| - |m\rangle\langle m|.
\end{align}
Starting from an energy eigenstate (say $|n\rangle$), a resonant $\pi/2$ pulse prepares the superposition $|\psi_0\rangle = \frac{1}{\sqrt{2}}\big(|n\rangle + e^{i\phi}|m\rangle\big)$, where $\phi$ is a controllable phase. Then the system evolves for a time $t$ under the
  pure-dephasing dynamics $\Phi_t$.
Finally, a second $\pi/2$ pulse with adjustable phase maps coherences back to population differences and is then followed by a population measurement. 
 Equivalently, one may regard
  the measurement as probing the expectation values of
  $\sigma_x^{(nm)}$ and $\sigma_y^{(nm)}$.

The Ramsey signal for a given phase setting is proportional to
\begin{equation}
  \langle \sigma_x^{(nm)}(t)\rangle
  = \Tr\big[\sigma_x^{(nm)} \rho_S(t)\big]
  = 2\,\mathrm{Re}\big[\rho_{nm}(t)\big],
\end{equation}
and similarly $\langle \sigma_y^{(nm)}(t)\rangle$ probes the imaginary part.
By scanning the phase $\varphi$ one reconstructs the complex coherence
$\rho_{nm}(t)$; the Ramsey fringe visibility is
\begin{equation}
  \mathcal{V}_{nm}(t) := \max_\varphi \langle \sigma_x^{(nm)}(t)\rangle
  = 2|\rho_{nm}(t)|.
\end{equation}
Hence, up to an irrelevant constant factor fixed by the initial preparation,
the visibility directly tracks the modulus of the coherence.
The decoherence rate can therefore be expressed in terms of the measured
Ramsey visibility as
\begin{equation}
  \gamma_{nm}(t)
  = -\frac{d}{dt}\ln|\rho_{nm}(t)|
  = -\frac{d}{dt}\ln \mathcal{V}_{nm}(t).
  \label{eq:ramsey-rate2}
\end{equation}
Equation~\eqref{eq:ramsey-rate2} gives $\gamma_{nm}(t)$ a direct operational
meaning as the decay rate of Ramsey fringe contrast. This parallels
Ref.~\cite{XuBlencowe}, where decoherence is extracted from gauge-invariant
interferometric visibility rather than from off-diagonal density-matrix
elements alone. Here Ramsey visibility plays the same role: it yields a
frame-invariant rate for one and the same physical interferometer, while local
Ramsey experiments tied to different phase references access the corresponding
local rates.
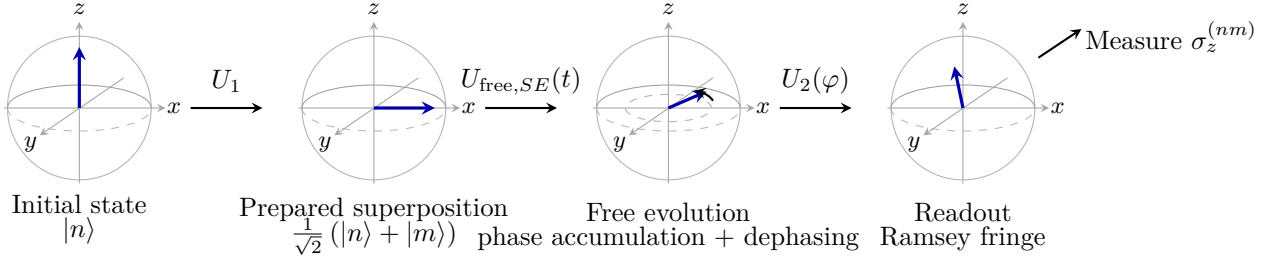
\begin{figure}[t]
\centering
\begin{tikzpicture}[scale=0.95,>=stealth]

\tikzset{
  sphere/.style={thin, gray!60},
  axis/.style={thin, gray!70, ->},
  equback/.style={thin, dashed, gray!55},
  equfront/.style={thin, gray!75},
  vec/.style={very thick, blue!65!black, ->},
  arr/.style={thick, ->},
  decoh/.style={thin, dashed, gray!60},
}

\newcommand{\blochsphere}[1]{
  \begin{scope}[xshift=#1 cm]
    \draw[sphere] (0,0) circle (1);
    \draw[equback] (-1,0) arc[start angle=180,end angle=360,x radius=1,y radius=0.32];
    \draw[equfront] (1,0) arc[start angle=0,end angle=180,x radius=1,y radius=0.32];
    \draw[axis] (0,-1.05) -- (0,1.2);          
    \draw[axis] (-1.05,0) -- (1.2,0);          
    \draw[axis] (0.62,0.42) -- (-0.55,-0.38);  
    \node[font=\small] at (0,1.38) {$z$};
    \node[font=\small] at (1.32,0) {$x$};
    \node[font=\small] at (-0.68,-0.5) {$y$};
  \end{scope}
}

\blochsphere{0}
\draw[vec] (0,0) -- (0,0.85);
\node[align=center] at (0,-1.55) {Initial state\\[-1mm] $|n\rangle$};

\draw[arr] (1.55,0) -- (2.55,0);
\node[font=\itshape] at (2.05,0.38) {$U_1$};

\blochsphere{4.1}
\draw[vec] (4.1,0) -- (4.95,0);
\node[align=center] at (4.1,-1.65)
  {Prepared superposition\\[-1mm]
   $\frac{1}{\sqrt{2}}\left(|n\rangle+|m\rangle\right)$};

\draw[arr] (5.65,0) -- (6.65,0);
\node[font=\itshape] at (6.15,0.38) {$U_{\mathrm{free},SE}(t)$};

\blochsphere{8.2}
\draw[decoh] (8.2,0) ellipse (0.6 and 0.19);
\draw[vec] (8.2,0) -- ({8.2+0.6*cos(28)},{0.19*sin(28)+0.30*sin(28)});
\draw[arr] ({8.2+0.62},{0.10}) arc[start angle=10,end angle=55,x radius=0.62,y radius=0.22];
\node[align=center] at (8.2,-1.65)
  {Free evolution\\[-1mm]
   phase accumulation $+$ dephasing};

\draw[arr] (9.75,0) -- (10.75,0);
\node[font=\itshape] at (10.25,0.38) {$U_2(\varphi)$};

\blochsphere{12.3}
\draw[vec] (12.3,0) -- (12.18,0.58);
\draw[arr] (13.35,0.7) -- (13.95,1.1);
\node[align=left] at (15.2,1)
  {Measure $\sigma_z^{(nm)}$};
\node[align=center] at (12.3,-1.65)
  {Readout\\[-1mm] Ramsey fringe};


\end{tikzpicture}
\caption{
Ramsey protocol on the effective qubit \(\{|n\rangle,|m\rangle\}\)
(\(\sigma_z^{(nm)}=|n\rangle\langle n|-|m\rangle\langle m|\)).
The pulses \(U_1\), \(U_2(\varphi)\) and the free evolution
\(U_{\mathrm{free},SE}(t)\) map the coherence \(\rho_{nm}\) onto a measurable
population imbalance; the transverse Bloch-vector length equals the Ramsey
visibility \(\mathcal V_{nm}=2|\rho_{nm}|\), from which the decoherence rate is
obtained.
}
\label{fig:ramsey_protocol_schematic}
\end{figure}

\subsection{Transformation of Ramsey protocol under QRF changes}
We now detail how an agent in frame $B$ describes a Ramsey interferometer 
physically implemented by agent $A$. In this subsection we restrict to the
factorized QRF action
$V_{SE}(g)=V_S(g)\otimes V_E(g)$.
This is the regime relevant for the frame-factor decomposition discussed
above. Since the Ramsey pulses and readout act trivially on \(E\), their
QRF-transformed versions also act trivially on \(E\), although they generally
become joint operations on \(S\) and the reference degrees of freedom.

\paragraph{Setup in frame $A$}
For concreteness, let us make explicit how the $\pi/2$ pulses and measurements implemented by agent
$A$ are described in frame $B$.  Denote by
\(U^{(A)}_{\mathrm{Ramsey}}(t)\) the Ramsey unitary on \(SBE\), corresponding to
the sequence ``\(\pi/2\)--free evolution--\(\pi/2\)'' and understood to act
trivially on the reference \(B\). Let \(\Pi_{S,\pm}\) be the projectors
associated with the outcomes of the final projective measurement of
\(\sigma_z^{(nm)}\) on the effective qubit
\(\{\ket n,\ket m\}\) of \(S\). The corresponding measurement effects, written
in the Heisenberg picture, are
\begin{equation}
  M^{(A)}_{\pm}(t)
  = U^{(A)\dagger}_{\mathrm{Ramsey}}(t)\,
    \big(\Pi^{(S)}_{\pm} \otimes \mathbb{I}_{B}\otimes \mathbb{I}_{E}\big)\,
    U^{(A)}_{\mathrm{Ramsey}}(t),
  \label{eq:Ramsey-POVM-A}
\end{equation}
so that the outcome probabilities in frame $A$ read $p^{(A)}_{\pm}(t)
  = \Tr_{SBE}\!\left[ M^{(A)}_{\pm}(t)\,\rho^{(A)}_{SBE}(0)\right]$.

In the $A$--frame, the two $\pi/2$ pulses appearing in $U^{(A)}_{\mathrm{Ramsey}}(t)$ act locally on
$S$ with respect to $A$'s phase reference. For instance, within the $\{|n\rangle,|m\rangle\}$ subspace
one may describe them as
\begin{equation}
  U^{(A)}_{1}
  = e^{-i\frac{\pi}{4}\sigma^{(nm)}_x} \otimes \mathbb{I}_{B}\otimes \mathbb{I}_{E},\qquad
  U^{(A)}_{2}(\varphi)
  = e^{-i\frac{\pi}{4}\big(\cos\varphi\,\sigma^{(nm)}_x + \sin\varphi\,\sigma^{(nm)}_y\big)}
    \otimes \mathbb{I}_{B}\otimes \mathbb{I}_{E},
\end{equation}
so that $U^{(A)}_{\mathrm{Ramsey}}(t)$ factorizes as $U^{(A)}_{\mathrm{Ramsey}}(t)
  = U^{(A)}_{2}(\varphi)\,U^{(A)}_{\mathrm{free},SE}(t)\,U^{(A)}_{1}$, where \(U^{(A)}_{\mathrm{free},SE}(t)\) acts nontrivially only on \(S\) and \(E\)
and trivially on \(B\). Indeed, $U^{(A)}_{\mathrm{free},SE}(t)$ is the joint \(SE\) unitary inducing the reduced
pure-dephasing channel \(\Phi^{(A)}_t\) on \(S\),
$\Phi^{(A)}_t(\rho_S)
=
\Tr_E\!\left[
U^{(A)}_{\mathrm{free},SE}(t)
(\rho_S\otimes\rho_E)
U^{(A)\dagger}_{\mathrm{free},SE}(t)
\right]$.

\paragraph{Description in frame $B$}
When the same physical setup is described in frame $B$, all operations implemented by $A$ are
conjugated by the QRF transformation $\hat S_{A\to B}$. In particular, the global Ramsey unitary
and the corresponding measurements become
\begin{align}
  U^{(B)}_{\mathrm{Ramsey}}(t)
  &= \hat S_{A\to B}\,U^{(A)}_{\mathrm{Ramsey}}(t)\hat S_{A\to B}^\dagger\, , \label{eq:URamseyC}\\
  M^{(B)}_{\pm}(t)
  &= \hat S_{A\to B}\,M^{(A)}_{\pm}(t)\,\hat S_{A\to B}^\dagger. \label{eq:Ramsey-POVM-C}
\end{align}
Substituting Eq.~\eqref{eq:Ramsey-POVM-A} into Eq.~\eqref{eq:Ramsey-POVM-C}, we find
\begin{equation}
  M^{(B)}_{\pm}(t)
  = U^{(B)\dagger}_{\mathrm{Ramsey}}(t)\,
    \big(\tilde\Pi^{(B)}_{\pm}  \otimes \mathbb{I}_E\big)\,
    U^{(B)}_{\mathrm{Ramsey}}(t),
\end{equation}
where
\begin{equation}
  \tilde\Pi^{(B)}_{\pm}
  := \hat S_{A\to B}\,
     \big(\Pi_{S,\pm}\otimes \mathbb{I}_B\otimes \mathbb{I}_E\big)\,
     \hat S_{A\to B}^\dagger
\end{equation}
are projectors on the joint $S+A$ system. 
Thus, from the viewpoint of an observer in frame $B$,
the $\pi/2$ pulses and readout that are \emph{local} operations on $S$ in frame $A$ are in general
described as \emph{joint} unitaries and measurements on $S$ and the reference $A$.
However, by cyclicity of the trace and unitarity of $\hat{S}_{A\to B}$, the 
outcome probabilities are identical:
\begin{equation}
  p^{(B)}_{\pm}(t)
  = \Tr_{SAE}\!\left[ M^{(B)}_{\pm}(t)\,\rho^{(B)}_{SAE}(0)\right]
  = \Tr_{SEB}\!\left[ M^{(A)}_{\pm}(t)\,\rho^{(A)}_{SEB}(0)\right]
  = p^{(A)}_{\pm}(t).
\end{equation}
 Consequently, the Ramsey visibility and
the associated operational decoherence rate of $A$'s interferometer,
$\mathcal{V}^{(\text{Ramsey }A)}_{nm}(t)$ and $\gamma^{(\text{Ramsey }A)}_{nm}(t)$, are frame-invariant, even
though the individual $\pi/2$ pulses and measurements appear differently in the two descriptions.

We finally note that the relation
\(
\gamma_{nm}^{(B)}(t)=\gamma_{nm}^{(A)}(t) - \frac{d}{dt}\ln|F_{nm}(t)|
\)
characterizes the free dephasing segment of the Ramsey sequence, during which
$S$ interacts with $E$ but not directly with the reference $B$, as well as that
the control pulses do not significantly disturb the reference degrees of freedom
entering $F_{nm}(t)$, e.g.\ because the reference is macroscopic or otherwise
robust against the applied controls.

\bibliographystyle{quantum}
\bibliography{bibliography}
\end{document}